\documentclass[12 pt]{article}
\usepackage{lmodern}
\usepackage[T1]{fontenc}
\usepackage{textcomp}
\usepackage[scr=boondoxo,frak=boondox,bb=boondox]{mathalfa}

\usepackage{epstopdf}
\usepackage{nameref}

\usepackage{amsmath}
\allowdisplaybreaks
\usepackage{amssymb}
\usepackage{amsthm}
\usepackage{ragged2e}
\usepackage{breqn}
\usepackage{graphicx}

\usepackage{hyperref}
\usepackage{appendix}
\usepackage{lipsum}
\usepackage{authblk}
\usepackage[top=2cm, bottom=2cm, left=2cm, right=2cm]{geometry}
\usepackage{fancyhdr}

\usepackage{enumitem}
\usepackage{float}
\usepackage{verbatim}
\usepackage{booktabs}
\renewenvironment{abstract}{
\hfill\begin{minipage}{0.95\textwidth}
\rule{\textwidth}{1pt}}
{\par\noindent\rule{\textwidth}{1pt}\end{minipage}}

\makeatletter
\renewcommand\@biblabel[1]{}
\makeatother

\makeatletter
\renewcommand\@maketitle{
\hfill
\begin{minipage}{0.95\textwidth}
\vskip 2em
\let\footnote\thanks 
{\LARGE \@title \par }
\vskip 1.5em
{\large \@author \par}
\end{minipage}
\vskip 1em \par
}
\makeatother

\newtheorem{theorem}{Theorem}
\newtheorem{lemma}{Lemma}

\title{\textbf{{\large ON APPROXIMATE LEAST SQUARES ESTIMATORS OF PARAMETERS OF ONE-DIMENSIONAL CHIRP SIGNAL}}}
\author[$\dagger$]{Rhythm Grover}
\author[$\dagger$,$\ddagger$]{Debasis Kundu}
\author[$\dagger$]{Amit Mitra}
\affil[$\dagger$]{Department of Mathematics, Indian Institute of Technology Kanpur, Kanpur - 208016, India}
\affil[$\ddagger$]{Corresponding author. Email: kundu@iitk.ac.in}
\date{}
\begin{document}
\maketitle

\begin{abstract}
\textbf{Abstract:} Chirp signals are quite common in many natural and man-made systems like audio signals, sonar, radar etc. Estimation of the unknown parameters of a signal is a fundamental problem in statistical signal processing. Recently, Kundu and Nandi \cite{2008} studied the asymptotic properties of least squares estimators of the unknown parameters of a simple chirp signal model under the assumption of stationary noise. In this paper, we propose periodogram-type estimators called the approximate least squares estimators to estimate the unknown parameters and study the asymptotic properties of these estimators under the same error assumptions. It is observed that the approximate least squares estimators are strongly consistent and asymptotically equivalent to the least squares estimators. Similar to the periodogram estimators, these estimators can also be used as initial guesses to find the least squares estimators of the unknown parameters. We perform some numerical simulations to see the performance of the proposed estimators and compare them with the least squares estimators and the estimators proposed by Lahiri et al., \cite{2013}. We have analysed two real data sets for illustrative purposes.
\\ \textbf{Key Words and Phrases: }Chirp signals, stationary, least squares estimators, approximate least squares estimators, consistent.
\end{abstract}

\section{Introduction}\label{section:1}
In this paper we consider the following multiple component \emph{chirp signal model}: \\
\begin{equation} \label{eq:1}
y(t) = \sum_{k=1}^{p}(A_k^0\cos(\alpha_k^0 t + \beta_k^0 t^2 ) + B_k^0\sin(\alpha_k^0 t + \beta_k^0 t^2)) + X(t);  \quad p \geq 1, \\
\end{equation}
\vspace{-\baselineskip} 
\justify
for \  $t = 1, \cdots, n$. Here $y(t)$ is the real valued signal observed at $t = 1, \cdots, n$. $A_k^0$s, $B_k^0$s are real valued \emph{amplitudes} and $\alpha_k^0$s, $\beta_k^0$s are the\emph{ frequencies} and the \emph{frequency rates}, respectively and $p$ is the number of components of the model. Here, $\{X(t)\}$ is a sequence of error random variables with mean zero and finite fourth moment. The explicit assumption on the error structure is provided in {Section~\ref{section:2}}.\\ \par
Unlike the sinusoidal signal, a chirp signal has a frequency that changes with time. These signals occur in many physical phenomena of interest in science and engineering. Chirp model has its roots in radar signal modelling and is used in various forms for modelling trajectories of moving objects. Also many estimation procedures have been proposed in the literature, for the estimation of the unknown parameters of chirp signals, which is of primary interest. See Bello \cite{1960}, Kelly \cite{1961}, Abatzoglou \cite{1986}, Djuric and Kay \cite{1990_1}, Peleg and Porat \cite{1991}, Shamsunder et al., \cite{1995}, Ikram et al., \cite{1997},  Besson et al., \cite{1999}, Saha and Kay \cite{2002}, Nandi and Kundu \cite{2004}, Kundu and Nandi \cite{2008} and references cited therein. For recent references, see Lahiri et al., \cite{2014},  \cite{2015} and Mazumder \cite{2016}. \\ \par

\begin{comment}
Least Squares Estimators are a reasonable choice for estimating the unknown parameters. The theoretical properties of the least squares estimators were first obtained by Nandi and Kundu (2004), under the assumption that the additive errors are i.i.d. complex valued random variables with mean zero and finite variance. They observed that the LSEs are consistent and asymptotically normally distributed. They also obtained the rates of convergence of the LSEs. Since in practice, the errors may not be independent, so to make the model more realistic, Kundu and Nandi (2008) considered the same model but a different error structure. They assumed stationarity of the error component to incorporate the dependence structure. Under this assumption, they obtained the LSEs of the unknown parameters and studied their asymptotic properties. In this case too, the LSEs are proved to be strongly consistent and asymptotically normally distributed and the rates of convergence of the amplitudes, frequency and frequency rate estimators are $n^{-1/2}$, $n^{-3/2}$ and $n^{-5/2}$ respectively. Also, it is observed that dispersion matrix of the asymptotic distribution of the LSEs turns out to be quite complicated. Using a number theoretic result of Vinogradov (1954), Lahiri et al., (2015) provided a simplified structure of this dispersion matrix. \\
\end{comment}
Least squares estimators (LSEs) are a reasonable choice for estimating the unknown parameters of a linear or a non-linear model. The theoretical properties of the LSEs for a chirp signal model, were first obtained by Nandi and Kundu \cite{2004} under the assumption that the additive errors are independently and identically distributed (i.i.d.) random variables with mean zero and finite variance. They proved that if the errors are i.i.d normal, the asymptotic variances attain the Cramer Rao lower bound. Since in practice, the errors may not be independent, so to make the model more realistic, Kundu and Nandi \cite{2008} assumed stationarity of the error component to incorporate the dependence structure and studied the properties of the LSEs of the same model. It is observed that dispersion matrix of the asymptotic distribution of the LSEs turns out to be quite complicated. Using a number theoretic result of Vinogradov \cite{1954}, Lahiri et al., \cite{2015} provided a simplified structure of this dispersion matrix. \\ \par
 
Although the LSEs have nice theroetical properties, finding the  least squares estimates is computationally quite demanding. For instance, for the sinusoidal model, it has been observed by Rice and Rosenblatt \cite{1988}, that the least squares surface has several local minima near the true parameter value (see Fig. 1, page 481) and due to this reason most of the iterative procedures, even when they converge, often converge to a local minimum rather than a global minimum. The same problem is observed for the chirp model. Thus a very good set of initial values are required for any iterative method to work. \\ \par
\begin{comment}
 The following least squares surface plot demonstrates this issue.\\
\graphicspath{C:/Users/DELL-PC/Desktop/Dropbox/Paper1}
\begin{figure}[H]
\begin{center}
\includegraphics[scale=0.15]{LSS.jpg}
%\includegraphics[scale=0.15]{"C:/Users/User/Dropbox/Paper1/LSS".jpg}
\caption{Least Squares Surface, n = 100, $A^0$ = 1, $B^0$ = 1.91, $\alpha^0$ = 2.5, $\beta^0$ = 0.1}
\label{fig:LSS_full}
\end{center}
\end{figure}
\justify
Figure (\ref{fig:LSS_full}) shows minus error sum of squares $Q_n (\hat{A},\hat{B},\alpha,\beta)$ of model~(\ref{eq:1}) with sample size, n = 100, and the true parameter values, $A^0$ = 1, $B^0$ = 1, $\alpha^0$ = 1.5, $\beta^0$ = 0.1, X(t) $\sim$ N(0,1). \par
\end{comment}

One of the most popular estimators for finding the initial values for the frequencies of the sinusoidal model are the periodogram estimators (PEs). These are obtained by maximizing the following periodogram function:  
\begin{equation}\label{eq:2}
I(\omega) = \frac{1}{n}\Bigg|\sum_{t=1}^{n} y(t)e^{-i(\omega t )}\Bigg| ^2
\end{equation}
at the Fourier frequencies, namely at $\displaystyle{\frac{\pi j}{n}}$; $j = 1, \cdots, n-1$. It has been proved that if the periodogram function $I(\omega)$ is maximised over the entire range $(0, \pi)$, the estimators obtained, called the approximate least squares estimators (ALSEs),  are consistent and asymptotically equivalent to the least squares estimators (see Whittle \cite{1952}, Walker \cite{1971}). In this paper, we study the behaviour of the periodogram-type estimators, of the unknown parameters of the chirp model and see how they compare with the corresponding least squares estimators theoretically.
Analogous to the periodogram function $I(\omega)$ for the sinusoidal model, a periodogram-type function for the chirp model can be defined as follows:
\begin{equation}\label{eq:3}
I(\alpha, \beta) = \frac{2}{n}\Bigg|\sum_{t=1}^{n} y(t)e^{-i(\alpha t + \beta t^2)}\Bigg| ^2.
\end{equation} 

Corresponding to the Fourier frequencies at which $I(\omega)$ is maximised for the sinusoidal model, it seems reasonable that for the chirp model, we maximise $I(\alpha, \beta)$ at $\displaystyle{\bigg(\frac{\pi j}{n}, \frac{\pi k}{n^2}\bigg)}$; $j = 1, \cdots, n-1$, $k = 1, \cdots, n^2-1$ to obtain the initial guesses for the frequency and frequency rate parameters, respectively. \\ \par
 %In this paper,  we propose periodogram-type estimators, called the approximate least squares estimators (ALSEs) to estimate the frequencies and frequency rates, the non-linear parameters of~{model (\ref{eq:1})}.
Consider the periodogram-like function defined in equation~(\ref{eq:3}), which can also be written as:
\begin{equation}\label{eq:4}
I(\alpha,\beta) = \frac{2}{n}\Bigg\{\Bigg(\sum\limits_{t=1}^{n}y(t)\cos(\alpha t + \beta t^2)\Bigg)^2  + \Bigg(\sum\limits_{t=1}^{n}y(t)\sin(\alpha t + \beta t^2)\Bigg)^2\Bigg\}. \\
\end{equation}
\justify
The ALSEs of $\alpha$ and $\beta$ are obtained by maximising $I(\alpha,\beta)$ with respect to $\alpha$ and $\beta$ simultaneously. Our primary focus is to estimate the non-linear parameters $\alpha$ and $\beta$, and once we estimate these parameters efficiently, the linear parameters $A$ and $B$ can be obtained by separable linear regression technique of Richards \cite{1961_2}. \\ \par

In this paper, we prove that the ALSEs are strongly consistent. As a matter of fact, the consistency of the ALSEs of the linear parameters $A$ and $B$ is obtained under slightly weaker conditions than that of the LSEs, as we do not require their parameter space to be bounded in this case. Also the rate of convergence of the ALSEs of the linear parameters is $n^{-1/2}$ and those of the frequency and frequency rate are $n^{-3/2}$ and $n^{-5/2}$, respectively. The convergence rates of ALSEs are thus same as that of their corresponding LSEs. We show that the asymptotic distribution of the ALSEs is equivalent to that of the LSEs. \\ \par

Recently, Lahiri et al., \cite{2013}, proposed an efficient algorithm to compute the estimators of the unknown parameters of the chirp model. We perform numerical simulations to compare the proposed ALSEs with the LSEs and the estimators obtained by the efficient algorithm. We observe that for most of the cases, although the LSEs provide the best results, the time taken by the ALSEs is comparatively less. Among the three estimators, the estimators computed using the efficient algorithm, takes the least amount of time, though the biases and MSEs increase as compared to the other two estimators. \\ \par
 
The rest of the paper is organised as follows. In section~\ref{section:2}, we prove the consistency of the ALSEs and their asymptotic equivalence to the LSEs. In section~\ref{section:3}, we discuss about the parameter estimation for the multiple component chirp model. In section~\ref{section:4} we present some simulation results and in section~\ref{section:5}, we analyze some real life data sets for illustrative purposes. Finally, in section~\ref{section:6} we conclude the paper. All the proofs have been provided in the appendices.

\section{Main Results for the One Component Chirp Model}\label{section:2} 
In this section, we study the asymptotic properties of the following one component chirp model: 
\begin{equation}\label{eq:5}
y(t) = A^0 \cos(\alpha^0 t + \beta^0 t^2) + B^0 \sin(\alpha^0 t + \beta^0 t^2) + X(t).
\end{equation}
\justify
We will use the following notations:
$\boldsymbol{\theta}$ = ($A$, $B$, $\alpha$, $\beta$), $\boldsymbol{\theta^0}$ = ($A^0$, $B^0$, $\alpha^0$, $\beta^0$), $\boldsymbol{\hat{\theta}}$ = ($\hat{A}$, $\hat{B}$, $\hat{\alpha}$, $\hat{\beta}$), the LSE of $\boldsymbol{\theta^0}$, and 
$\boldsymbol{\tilde{\theta}}$ = $(\tilde{A}, \tilde{B}, \tilde{\alpha}, \tilde{\beta})$, the ALSE of $\boldsymbol{\theta^0}$.
The following assumptions are made on the error component $X(t)$ of model~(\ref{eq:5}):
\justify
$\textsc{Assumption 1.}$ Let $Z$ be the set of integers. $\{X(t)\}$ is a stationary linear process with the following form:
\begin{equation}\label{eq:6}
X(t) = \sum_{j = -\infty}^{\infty} a(j)e(t-j),
\end{equation}
where $\{e(t); t \in Z\}$ is a sequence of  i.i.d random variables with $E(e(t)) = 0$, $V(e(t)) = \sigma^2$, and $a(j)$s are real constants such that 
\begin{equation}\label{eq:7}
\sum\limits_{j= - \infty}^{\infty}|a(j)| < \infty.
\end{equation} 
\justify
This is a standard assumption for a stationary linear process. Any finite dimensional stationary MA, AR or ARMA process can be represented as~(\ref{eq:6}) when the coefficients $a(j)$s satisfy condition~(\ref{eq:7}) and hence this covers a large class of stationary random variables. \\
\justify
Let $\tilde{A}$, $\tilde{B}$, $\tilde{\alpha}$ and $\tilde{\beta}$, be the ALSEs of $A^0$, $B^0$, $\alpha^0$ and $\beta^0$, respectively. First we find $\tilde{\alpha}$ and $\tilde{\beta}$ by maximising $I(\alpha, \beta)$, as defined in~(\ref{eq:4}) with respect to $\alpha$ and $\beta$ and once we obtain $\tilde{\alpha}$ and $\tilde{\beta}$, the ALSEs of the linear parameters $A$ and $B$ can be obtained as follows: \\ \\
\begin{flalign}\label{eq:8}
\begin{split}
\tilde{A} = \frac{2}{n} \sum_{t=1}^{n} y(t)\cos(\tilde{\alpha} t + \tilde{\beta} t^2)\textmd{ and }
\tilde{B} = \frac{2}{n} \sum_{t=1}^{n} y(t)\sin(\tilde{\alpha} t + \tilde{\beta} t^2). \\ 
\end{split}
\end{flalign}
In the following two theorems, we state the consistency of the ALSE, $\boldsymbol{\tilde{\theta}}$. \\

\begin{theorem}\label{theorem:1}
Let ($\alpha^0$, $\beta^0$) be an interior point of [0, $\pi$] $\times$ [0, $\pi$]. If \{X(t)\} satisfies Assumption 1, then the ALSEs  $\tilde{\alpha}$ and $\tilde{\beta}$ are strongly consistent estimators of $\alpha^0$ and $\beta^0$, respectively.
\end{theorem}

\begin{proof}
See \nameref{appendix:A}.
\end{proof}

\begin{theorem}\label{theorem:2}
Under the conditions of Theorem~\ref{theorem:1}, the ALSEs $\tilde{A}$ and $\tilde{B}$ of the linear parameters $A^0$ and $B^0$ are strongly consistent estimators.
\end{theorem}

\begin{proof}
See \nameref{appendix:A}.
\end{proof}

\justify
It has been observed in the following theorem that ALSEs have the same distribution as the LSEs asymptotically. \\
\begin{theorem}\label{theorem:3}
Under the Assumption 1, the limiting distribution of $(\boldsymbol{\tilde{\theta}} - \boldsymbol{\theta^0})\mathbf{D}^{-1}$ is same as that of $(\boldsymbol{\hat{\theta}} - \boldsymbol{\theta^0})\mathbf{D}^{-1} \  as \ n \rightarrow \infty$, where $\boldsymbol{\tilde{\theta}}$ is the ALSE of $\boldsymbol{\theta^0}$ and $\boldsymbol{\hat{\theta}}$ is the LSE of $\boldsymbol{\theta^0}$ and $\mathbf{D}$ = $diag(\frac{1}{\sqrt{n}} , \frac{1}{\sqrt{n}}  , \frac{1}{n\sqrt{n}} , \frac{1}{n^2\sqrt{n}}).$
\end{theorem}
\begin{proof}
See \nameref{appendix:B}.
\end{proof}
\justify
\section{Main Results for the Multiple Component Chirp Model}\label{section:3}
In this section, we consider a chirp signal model with multiple components. Mathematically, a multiple-component chirp model is given by: \\
\begin{equation}\label{eq:9}
y(t) = \sum_{k=1}^{p}(A_{k}^{0} \cos(\alpha_{k}^{0}t + \beta_{k}^{0} t^2) +B_{k}^{0} \sin(\alpha_{k}^{0}t + \beta_{k}^{0} t^2)) + X(t), \ \ \ p > 1,\ t = 1, 2, \cdots, n,
\end{equation}
 where  $y(t)$ is the real valued signal observed at $t$ = 1, 2 , $\cdots$, $n$, $A_{k}^{0}s$ and $B_{k}^{0}s$ are the amplitudes and $\alpha_{k}^{0}s$ and $\beta_{k}^{0}s$ are the frequencies and frequency rates, respectively for $k$ = 1, 2, $\cdots$, $p$. \\ \par
\begin{comment}
Now we discuss the asymptotic properties of the ALSEs of the parameters of a chirp model with multiple components.
Mathematically, a multiple-component chirp model is given by: \\
\begin{equation}\label{eq:11}
y(t) = \sum_{k=1}^{p}(A_{k}^{0} \cos(\alpha_{k}^{0}t + \beta_{k}^{0} t^2) +B_{k}^{0} \sin(\alpha_{k}^{0}t + \beta_{k}^{0} t^2)) + X(t) \ \ \ p > 1,\ t = 1, 2, \cdots, n
\end{equation}
 where  $y(t)$ is the real valued signal observed at $t$ = 1, 2 , $\cdots$, $n$, $A_{k}^{0}s$ and $B_{k}^{0}s$ are the amplitudes and $\alpha_{k}^{0}s$ and $\beta_{k}^{0}s$ are the frequencies and frequency rates respectively for $k$ = 1, 2, $\cdots$, $p$. Our problem again, is to estimate the unknown parameters, that is the amplitudes, the frequencies and the frequency rates by approximate least squares method of estimation and then to investigate properties of these estimators.\\ \\
\end{comment}
\justify
To estimate the unknown parameters, we propose a sequential procedure to find the ALSEs. This method reduces the computational complexity of the estimators significantly without compromising on their efficiency. Following is the algorithm to find the ALSEs through sequential method: \\
\begin{comment}
To find the ALSEs of the unknown parameters of the above model, we propose a sequential method which reduces the computational complexity of the estimators without compromising on their efficiency. Following is the algorithm to find the ALSEs through sequential method: \\
\end{comment}
\justify
Step 1: Compute $\tilde{\alpha_{1}}$ and $\tilde{\beta_{1}}$ by maximizing the periodogram-like function
\begin{equation}\label{eq:10}
I_{1}(\alpha,\beta)  = \frac{1}{n}\bigg\{\bigg(\sum_{t=1}^{n}y(t)\cos(\alpha t + \beta t^2)\bigg)^2 + \bigg(\sum_{t=1}^{n}y(t)\sin(\alpha t + \beta t^2)\bigg)^2\bigg\}.
\end{equation}
Then the linear parameter estimates can be obtained by substituting  $\tilde{\alpha_{1}}$ and $\tilde{\beta_{1}}$ in~(\ref{eq:8}).
Thus
\begin{flalign*}
\begin{split}
\tilde{A_1} = \frac{2}{n} \sum_{t=1}^{n} y(t)\cos(\tilde{\alpha_1} t + \tilde{\beta_1} t^2)\textmd{ and }
\tilde{B_1} = \frac{2}{n} \sum_{t=1}^{n} y(t)\sin(\tilde{\alpha_1} t + \tilde{\beta_1} t^2). \\ 
\end{split}
\end{flalign*}
Step 2: Now we have the estimates of the parameters of the first component of the observed signal. We subtract the contribution of the first component from the original signal $y(t)$ to remove the effect of the first component and obtain new data, say $$ y^1(t) = y(t) - \tilde{A_{1}}\cos(\tilde{\alpha_{1}} t + \tilde{\beta_{1}} t^2) - \tilde{B_{1}}\sin(\tilde{\alpha_{1}} t + \tilde{\beta_{1}} t^2),\ \ \ \ t = 1, 2, \cdots, n.$$
Step 3: Now compute $\tilde{\alpha_{2}}$ and $\tilde{\beta_{2}}$ by maximizing $I_{2}(\alpha,\beta)$ which is obtained by replacing the original data vector by the new data vector in~(\ref{eq:10}) and $\tilde{A_{2}}$ and $\tilde{B_{2}}$ by substituting $\tilde{\alpha_{2}}$ and $\tilde{\beta_{2}}$ in~(\ref{eq:8}). \\ \\
Step 4: Continue the process upto $p$-steps. \\ \\
Note that we use the following notation: the parameter vector $\boldsymbol{\theta_{k}} = (A_{k}, B_{k}, \alpha_{k}, \beta{k})^T$ and  the true parameter vector $\boldsymbol{\theta_{k}^0} = (A_{k}^0, B_{k}^0, \alpha_{k}^0, \beta{k}^0)^T$ for all $k$ = 1, 2, $\cdots$, $p$ and the parameter space $\boldsymbol{\Theta} = (-\infty,\infty) \times (-\infty,\infty) \times [0,\pi] \times[0,\pi].$ \\ \\
Next to establish the asymptotic properties of these estimators, we further make the following model assumptions: \\ \\
%1. The error random variables $\{X(t)\}$ satisfy the same condition as for the one component model. See Assumption 1. \\ \\
%2. $\theta_{k}^{0}$ is an interior point of $\Theta$ $\forall$ $k = 1(1)p$ and the frequencies $\alpha_{k}^0s$ are distinct and so are the frequency rates $\beta_{k}^0s$. \\ \\
$\textsc{Assumption 2.}$ $\boldsymbol{\theta_{k}^{0}}$ is an interior point of $\boldsymbol{\Theta}$ $\forall$ $k = 1(1)p$ and the frequencies $\alpha_{k}^0s$ and the frequency rates $\beta_{k}^0s$ are such that $(\alpha_i^0, \beta_i^0) \neq (\alpha_j^0, \beta_j^0)$ $\forall i \neq j$. \\ \\
$\textsc{Assumption 3.}$ $A_k^0$s and $B_k^0$s satisfy the following relationship: \\
$$ \infty > {A_{1}^{0}}^2 + {B_{1}^{0}}^2 > {A_{2}^{0}}^2 + {B_{2}^{0}}^2 > \cdots > {A_{p}^{0}}^2 + {B_{p}^{0}}^2 > 0. $$
\justify
In the following theorems we prove that the ALSEs obtained by the sequential method described above are strongly consistent.
\begin{theorem}\label{theorem:4}
Under the assumptions 1, 2 and 3, $\tilde{A_1}, \tilde{B_1}, \tilde{\alpha_1}$ and $\tilde{\beta_1}$ are strongly consistent estimators of $A_{1}^{0}, B_{1}^{0},  \alpha_{1}^{0}$ and $\beta_{1}^{0}$ respectively, that is, $\boldsymbol{\tilde{\theta_1}} \xrightarrow{a.s.} \boldsymbol{\theta_1^0}$ as $n \rightarrow \infty$.
\end{theorem}
\begin{proof}
See \nameref{appendix:C}.
\end{proof}
\begin{theorem}\label{theorem:5}
If the assumptions 1, 2 and 3 are satisfied and p $\geqslant$ 2, $\tilde{A_2}, \tilde{B_2}, \tilde{\alpha_2}$ and $\tilde{\beta_2}$ are strongly consistent estimators of $A_{2}^{0}, B_{2}^{0},  \alpha_{2}^{0}$ and $\beta_{2}^{0}$, respectively, that is, $\boldsymbol{\tilde{\theta_2}} \xrightarrow{a.s.} \boldsymbol{\theta_2^0}$ as $n \rightarrow \infty$.
\end{theorem}
\begin{proof}
See \nameref{appendix:C}.
\end{proof}
\justify
The result obtained in the above theorem can be extended upto the $p$-th step. Thus for any $k$ $\leqslant$ $p$, the ALSEs obtained at the $k$-th step are strongly consistent.
\begin{theorem}\label{theorem:6}
If the assumptions 1, 2 and 3 are satisfied, and if $\tilde{A_k}$,  $\tilde{B_k}$, $\tilde{\alpha_k}$ and $\tilde{\beta_k}$ are the estimators obtained at the $k$-th step, and k $>$ p then $\tilde{A_k}$ $\xrightarrow{a.s}$ 0 and $\tilde{B_k}$ $\xrightarrow{a.s}$ 0 as $n \rightarrow \infty$.
\end{theorem}
\begin{proof}
See \nameref{appendix:C}.
\end{proof}
\justify
Lahiri et al., \cite{2015} proved that the ordinary LSEs of the unknown parameters of the $p$-component chirp model have the following asymptotic distribution: \\
$$((\boldsymbol{\hat{\theta_1}} - \boldsymbol{\theta_1^0})\mathbf{D}^{-1}, \cdots, (\boldsymbol{\hat{\theta_p}} - \boldsymbol{\theta_p^0})\mathbf{D}^{-1}) \xrightarrow{d} N_{4p}(0 , 2c\sigma^2 \boldsymbol{\Sigma}(\boldsymbol{\theta^0})).$$  \\
$\textmd{Here, }\quad \mathbf{D} = diag(\frac{1}{\sqrt{n}}, \frac{1}{\sqrt{n}}, \frac{1}{n\sqrt{n}}, \frac{1}{n^2\sqrt{n}}), \quad c = \sum\limits_{j = -\infty}^{\infty} a(j)^2 \ \  \textmd{and }$ \\ \\
$$\boldsymbol{\Sigma}(\boldsymbol{\theta^{0}}) = \left[\begin{array}{cccc}\boldsymbol{\Sigma_1} & 0 & \cdots & 0 \\0 & \boldsymbol{\Sigma_2} & \cdots & 0 \\ \vdots & \vdots & \ddots & \vdots \\0 & 0 & \cdots & \boldsymbol{\Sigma_p}\end{array}\right],$$ \\
where, \\
\begin{equation}\label{eq:11}
\boldsymbol{\Sigma_k} = \frac{2}{{A_{k}^0}^2 + {B_{k}^0}^2}\left[\begin{array}{cccc}\frac{1}{2}\bigg({A_{k}^0}^2 + 9{B_{k}^0}^2\bigg) & -4A_k^0B_k^0 & -18B_k^0 & 15B_k^0 \\-4A_k^0B_k^0  & \frac{1}{2}\bigg({9A_{k}^0}^2 +{B_{k}^0}^2\bigg) & 18A_k^0 & -15A_k^0 \\ -18B_k^0 & 18A_k^0 & 96 & -90 \\ 15B_k^0 & -15A_k^0 & -90 & 90 \end{array}\right],\ \forall \ \  k = 1, \cdots, p.\\ \\
\end{equation}
Also, note that  
\begin{equation}\label{eq:12}
\boldsymbol{\Sigma_k}^{-1} = \left[\begin{array}{cccc}1 & 0 & \frac{B_k^0}{2} & \frac{B_k^0}{3} \\0  & 1 & -\frac{A_k^0}{2} & -\frac{A_k^0}{3} \\ \frac{B_k^0}{2} & -\frac{A_k^0}{2} & \frac{{A_k^0}^2 + {B_k^0}^2}{3} & \frac{{A_k^0}^2 + {B_k^0}^2}{4} \\ \frac{B_k^0}{3} & -\frac{A_k^0}{3} & \frac{{A_k^0}^2 + {B_k^0}^2}{4} & \frac{{A_k^0}^2 + {B_k^0}^2}{5} \end{array}\right]. \\
\end{equation}
\justify
\begin{comment}
Lahiri et al., (2015) proposed a sequential procedure to estimate the parameters of the $p$-component model. In this method, they find the LSEs of the unknown parameters of each component sequentially, thus reducing the problem to two dimensional optimization at each step. They state that if the following conjecture, see Montgomery (1990), is satisfied the proposed sequential least squares estimates have the same asymptotic distribution as that of the ordinary LSEs as mentioned above.
\justify
Conjecture: \\ \\
If $\theta_1, \theta_2, \theta_1', \theta_2' \in (0 , \pi)$, then except for countable number of points
$$\lim_{n \rightarrow \infty} \frac{1}{\sqrt{n}n^k}\sum_{t=1}^{n} t^k \cos(\theta_1 t + \theta_2 t^2)\cos(\theta_1' t + \theta_2' t^2) = 0; \quad t = 0, 1, 2. $$
In addition, if $\theta_2 = \theta_2'$, 
$$\lim_{n \rightarrow \infty} \frac{1}{\sqrt{n}n^k}\sum_{t=1}^{n} t^k \cos(\theta_1 t + \theta_2 t^2)\cos(\theta_1' t + \theta_2' t^2) = 0; \quad k= 0, 1, 2. $$
$$\lim_{n \rightarrow \infty} \frac{1}{\sqrt{n}n^k}\sum_{t=1}^{n} t^k \sin(\theta_1 t + \theta_2 t^2)\sin(\theta_1' t + \theta_2' t^2) = 0; \quad k = 0, 1, 2. $$ \\
Here we propose an analogous sequential procedure, and at each step we compute the ALSEs. Now through the following theorem we prove that the ALSEs obtained at each step are asymptotically equivalent to the corresponding LSEs.
\end{comment}
We have the following result regarding the asymptotic distribution of the ALSEs.
\begin{theorem}\label{theorem:7}
Under the assumptions 1, 2, and 3, the asymptotic distribution of $(\boldsymbol{\tilde{\theta_k}} - \boldsymbol{\theta_k^0}) \mathbf{D}^{-1}$ is equivalent to the asymptotic distribution of $(\boldsymbol{\hat{\theta_k}} - \boldsymbol{\theta_k^0}) \mathbf{D}^{-1}$, for all $k = 1, \cdots, p$, where $\boldsymbol{\tilde{\theta_k}}$ is the ALSE and $\boldsymbol{\hat{\theta_k}}$ is the LSE of the unknown parameter vector associated with the $k$-th component of the $p$ component model.
\end{theorem}
\begin{proof}
See \nameref{appendix:D}.
\end{proof}
\section{Numerical Experiments}\label{section:4}
In this section, we present simulation studies for one component and two component chirp models. We first consider the following one component chirp model: \\
$$y(t) = A^0 \cos(\alpha^0 t + \beta^0 t^2) + B^0 \sin(\alpha^0 t + \beta^0 t^2) + X(t),$$\\
with the true parameter values $A^0 = 2.93, B^0 = 1.91, \alpha^0 = 2.5,$ and $\beta^0 = 0.1$ and $X(t)$ is an MA(1) process, that is $X(t) = e(t) + \rho e(t-1)$, with $\rho = 0.5$ and $e(t)$s are i.i.d. normal random variables with mean zero and variance $\sigma^2$. For simulations we consider different $\sigma^2$: 0.1, 0.5 and 1. The different sample sizes we use are $n = 250$, $n = 500$ and $n = 1000$ and for each $n$ we replicate the process, that is generate the data and obtain the estimates 1000 times. We estimate the parameters by the least squares estimation method, the approximate least squares estimation method and using the efficient algorithm as proposed by Lahiri et al., \cite{2013}. \\ \par
For the LSEs, we first minimize the error sum of squares function with respect to $\alpha$ and $\beta$ using the Nelder and Mead method of optimization (using optim function in the R Stats Package). For the initial values, it is intuitive to minimize the function over the grid $\displaystyle{(\frac{\pi j}{n},\frac{\pi k}{n^2})}$, $j = 1, \cdots, n; \  k = 1, \cdots, n^2$ analogous to what is suggested by Rice and Rosenblatt \cite{1988} for the sinusoidal model.
For the ALSEs, we maximize the periodogram-like function $I(\alpha,\beta)$, as defined in~(\ref{eq:4}), again using the Nelder and Mead method and the starting values are obtained by maximizing $I(\alpha,\beta)$ on grid points as used for the corresponding LSEs.\\ \par
\begin{table}[H]
\centering
\resizebox{0.75\textwidth}{!}{\begin{tabular}{|c|c|c|c|c|c|c|c|}
\hline
\multicolumn{2}{|c|}{Non-linear Parameters}  & $\alpha$ & $\beta$ & $\alpha$ & $\beta$ & $\alpha$ & $\beta$ \\
\hline  \multicolumn{2}{|c|}{True values}  & 2.5 & 0.1 & 2.5 & 0.1 & 2.5 & 0.1\\
\hline
$\sigma^2$&& \multicolumn{2}{|c|}{LSEs} & \multicolumn{2}{|c|}{ALSEs} & \multicolumn{2}{|c|}{Efficient Algorithm}\\
\hline 0.1  & Time (s)     &  17.2249    &             &  15.3530   &           &  2.2450   &           \\
\hline      & Average  &  2.5000     &  0.1000     &  2.4967    & 0.1000    &  2.4801   &   0.0997  \\
\hline      & Bias     &  4.15e-06   & -9.36e-09   & -3.26e-03  & 9.28e-06  & -1.99e-02 &  -2.54e-04\\
\hline      & MSE      &  1.80e-07   &  2.68e-12   &  1.10e-05  & 9.12e-11  &  8.22e-03 &   1.67e-07\\
\hline      & Avar     &  1.26e-07   &  1.88e-12   &  1.26e-07  & 1.88e-12  &  1.26e-07 &   1.88e-12\\
\hline      &          &             &             &            &           &           &           \\
\hline 0.5  & Time (s)     &  17.4460    &             &  13.4589   &           & 2.2410    &           \\
\hline      & Average  &  2.5000     &  0.1000     &  2.4968    & 0.1000    & 2.5063    &  0.0999   \\
\hline      & Bias     &  4.63e-05   & -1.97e-07   & -3.19e-03  & 8.98e-06  & 6.27e-03  & -1.29e-04 \\
\hline      & MSE      &  8.84e-07   &  1.33e-11   &  1.18e-05  & 1.05e-10  & 1.96e-02  &  3.29e-07 \\
\hline      & Avar     &  6.28e-07   &  9.42e-12   &  6.28e-07  & 9.42e-12  & 6.28e-07  &  9.42e-12 \\
\hline      &          &             &             &            &           &           &           \\
\hline  1   & Time (s)     &  18.2910    &             &  14.2689   &           & 2.4050    &           \\
\hline      & Average  &  2.5000     &  0.1000     &  2.4968    & 0.1000    & 2.5021    &  0.0999   \\
\hline      & Bias     & -7.89e-06   &  2.49e-08   & -3.25e-03  & 9.23e-06  & 2.13e-03  &  -1.27e-04\\
\hline      & MSE      &  1.89e-06   &  2.94e-11   &  1.40e-05  & 1.37e-10  & 4.62e-02  &  6.47e-07 \\
\hline      & Avar     &  1.26e-06   &  1.88e-11   &  1.26e-06  & 1.88e-11  & 1.26e-06  &  1.88e-11 \\
\hline
\end{tabular}}
\caption{Estimates of one component model when sample size is 250}
\label{table:1}
\end{table}
Tables~\ref{table:1},~\ref{table:2} and~\ref{table:3} provide the results, averaged over 1000 simulation runs, we obtain for the one component model.
In these tables, we observe that, the ALSEs have very small bias in the absolute value. The MSEs of the LSEs are very close to their asymptotic variances and the MSEs of the ALSEs also get very close to those of LSEs as $n$ increases and hence to the theoretical asymptotic variances of the LSEs, showing that they are asymptotically equivalent. Also when we increase the sample size, the MSEs of both the estimators decrease showing that they are consistent. We observe that the estimators obtained by the Efficient Algorithm are close to the true values but the bias and the MSEs are not as small as compared with the other two estimators. However the time taken to compute the estimates by the Efficient Algorithm is much less than the time taken by the ALSEs and the LSEs.

\begin{table}[H]
\centering
\resizebox{0.75\textwidth}{!}{\begin{tabular}{|c|c|c|c|c|c|c|c|}
\hline \multicolumn{2}{|c|}{Non-linear Parameters} & $\alpha$ & $\beta$ & $\alpha$ & $\beta$ & $\alpha$ & $\beta$\\
\hline \multicolumn{2}{|c|}{True values} & 2.5 & 0.1 & 2.5 & 0.1 & 2.5 & 0.1\\
\hline $\sigma^2$&& \multicolumn{2}{|c|}{LSEs} & \multicolumn{2}{|c|}{ALSEs} & \multicolumn{2}{|c|}{Efficient Algorithm}\\
\hline 0.1 &  Time (s)        &  30.1289    &              &   25.6280    &            &  4.3600    &              \\
\hline & Average     &  2.5000     &   0.1000     &   2.4993     &  0.1000    &  2.5304    &   0.1001     \\
\hline & Bias        &  -1.19e-05  &   2.56e-08   &  -6.79e-04   &  1.77e-06  &  3.04e-02  &   7.11e-05   \\
\hline & MSE         &  2.13e-08   &   8.09e-14   &   4.96e-07   &  3.26e-12  &  9.89e-03  &   3.11e-08   \\
\hline & Avar        &  1.57e-08   &   5.89e-14   &   1.57e-08   &  5.89e-14  &  1.57e-08  &   5.89e-14   \\
\hline    &          &             &              &              &            &            &              \\
\hline  0.5 & Time (s)        &  34.6510    &              &   30.3080    &            &  5.2440    &              \\
\hline & Average     &  2.5000     &   0.1000     &   2.4994     &  0.1000    &  2.5227    &   0.1000     \\
\hline & Bias        &  1.04e-05   &  -1.44e-08   &  -6.47e-04   &  1.71e-06  &  2.27e-02  &   4.54e-05   \\
\hline & MSE         &  1.21e-07   &   4.45e-13   &   6.12e-07   &  3.63e-12  &  2.35e-02  &   6.69e-08   \\
\hline & Avar        &  7.85e-08   &   2.94e-13   &   7.85e-08   &  2.94e-13  &  7.85e-08  &   2.94e-13   \\
\hline      &        &             &              &              &            &            &              \\
\hline 1 & Time (s)        &  32.2790    &              &   26.5430    &            &  4.4199    &              \\
\hline & Average     &  2.5000     &   0.1000     &   2.4993     &  0.1000    &  2.5189    &   0.1000     \\
\hline & Bias        &  -1.61e-05  &   2.01e-08   &  -6.77e-04   &  1.75e-06  &  1.89e-02  &   3.15e-05   \\
\hline & MSE         &  2.18e-07   &   8.08e-13   &   8.04e-07   &  4.33e-12  &  7.32e-03  &   2.31e-08   \\
\hline & Avar        &  1.57e-07   &   5.89e-13   &   1.57e-07   &  5.89e-13  &  1.57e-07  &   5.89e-13   \\
\hline
\end{tabular}}
\caption{Estimates of one component model when sample size is 500}
\label{table:2}
\end{table}

\begin{table}[H]
\centering
\resizebox{0.75\textwidth}{!}{\begin{tabular}{|c|c|c|c|c|c|c|c|c|}
\hline
\multicolumn{2}{|c|}{Non-linear Parameters} & $\alpha$ & $\beta$ & $\alpha$ & $\beta$ & $\alpha$ & $\beta$\\
\hline  \multicolumn{2}{|c|}{True values} & 2.5 & 0.1  & 2.5 & 0.1 & 2.5 & 0.1\\
\hline
$\sigma^2$&& \multicolumn{2}{|c|}{LSEs} & \multicolumn{2}{|c|}{ALSEs} & \multicolumn{2}{|c|}{Efficient Algorithm}\\
\hline 0.1 &  Time (s)            &    67.1180   &              &   62.3369    &                &   9.5720   &               \\
\hline     &  Average         &    2.5000    &   0.1000     &   2.5002     &     0.1000     &   2.4984   &      0.1000    \\
\hline     &  Bias            &    8.16e-07  &  -9.15e-10   &   1.86e-04   &    -9.30e-08   &  -1.65e-03 &     -5.34e-06\\
\hline     &  MSE             &    2.95e-09  &   2.85e-15   &   3.87e-08   &     1.21e-14   &   1.79e-03 &      1.46e-09\\
\hline     &  Avar            &    1.96e-09  &   1.84e-15   &   1.96e-09   &     1.84e-15   &   1.96e-09 &      1.84e-15\\
\hline     &            &                         &              &              &                &            &                \\
\hline 0.5 &  Time (s)            &   61.2009    &              &  56.3849     &                &  8.2260    &              \\
\hline     &  Average         &   2.5000     &  0.1000      &  2.5002      &    0.1000      &  2.4981    &     0.1000   \\
\hline     &  Bias            &   1.80e-06   & -1.67e-09    &  1.86e-04    &   -9.24e-08    & -1.87e-03  &    -4.21e-06\\
\hline     &  MSE             &   1.57e-08   &  1.55e-14    &  5.40e-08    &    2.60e-14    &  1.58e-03  &     1.32e-09\\
\hline     &  Avar            &   9.81e-09   &  9.20e-15    &  9.81e-09    &    9.20e-15    &  9.81e-09  &     9.20e-15\\
\hline     &              &                   &              &              &                &            &              \\
\hline  1  &  Time (s)            &   62.5589    &              &  56.3840     &                &  8.2129    &             \\
\hline     &  Average         &   2.5000     &  0.1000      &  2.5002      &    0.1000      &  2.4948    &     0.1000\\
\hline     &  Bias            &   3.32e-06   & -8.67e-10    &  1.88e-04    &   -9.19e-08    & -5.20e-03  &    -6.51e-06\\
\hline     &  MSE             &   3.10e-08   &  2.95e-14    &  7.41e-08    &    4.22e-14    &  1.40e-03  &     1.13e-09\\
\hline     &  Avar            &   1.96e-08   &  1.84e-14    &  1.96e-08    &    1.84e-14    &  1.96e-08  &     1.84e-14\\
\hline
\end{tabular}}
\caption{Estimates of one component model when sample size is 1000}
\label{table:3}
\end{table}

\justify
We also perform simulations for the following two component model using the proposed sequential estimators: $$y(t) = A_1^0 \cos(\alpha_1^0 t + \beta_1^0 t^2) + B_1^0 \sin(\alpha_1^0 t + \beta_1^0 t^2) + A_2^0 \cos(\alpha_2^0 t + \beta_2^0 t^2) + B_2^0 \sin(\alpha_2^0 t + \beta_2^0 t^2) + X(t).$$
\justify
For simulation, we take the true values as $A_1^0$ = 2, $B_1^0$ = 1.75, $\alpha_1^0$ = 1.5, $\beta_1^0$ = 0.1, $A_2^0$ = 3, $B_2^0$ = 2.25, $\alpha_2^0$ = 2.5 and $\beta_2^0$ = 0.2, and compute both the LSEs and the ALSEs of all the unknown parameters, sequentially. The error structure is same as that for one component simulation study.
\begin{table}[H]
\centering
\resizebox{0.75\textwidth}{!}{\begin{tabular}{|c|c|c|c|c|c|c|c|}
\hline
 \multicolumn{2}{|c|}{Non-linear Parameters} & $\alpha_1$ & $\beta_1$&  $\alpha_1$ & $\beta_1$& $\alpha_1$ & $\beta_1$ \\
\hline $\sigma^2$ & & \multicolumn{2}{|c|}{LSEs} & \multicolumn{2}{|c|}{ALSEs} & \multicolumn{2}{|c|}{Efficient Algorithm}\\
\hline        \multicolumn{2}{|c|}{True values} & 1.5 & 0.1 &  1.5 & 0.1 & 1.5 & 0.1 \\
\hline   0.1 &   Time (s)            & 31.918    &        &  25.5269  &          &  4.5620 & \\
\hline       &   Average         &  1.5074    &  0.1000     & 1.5044    &  0.1000     &  1.4711    & 0.1001   \\
\hline       &   Bias            &  7.43e-03  & -2.57e-05   & 4.40e-03  & -1.59e-05   & -2.89e-02  & 7.77e-05 \\
\hline       &   MSE             &  5.58e-05  &  6.70e-10   & 2.00e-05  &  2.62e-10   &  1.63e-02  & 2.25e-07 \\
\hline       &   Avar            &  1.09e-07  &  1.64e-12   & 1.09e-07  &  1.64e-12   &  1.09e-07  & 1.64e-12 \\
\hline  0.5  &   Time (s)              & 32.3660   &        &  26.2630  &          & 4.5389 & \\
\hline       &   Average         &  1.5075    &  0.1000    &  1.5045    &  0.1000     &  1.4832     & 0.1001   \\
\hline       &   Bias            &  7.51e-03  & -2.60e-05  &  4.48e-03  & -1.62e-05   & -1.68e-02   & 1.18e-04 \\
\hline       &   MSE             &  5.80e-05  &  7.04e-10  &  2.29e-05  &  3.15e-10   &  2.41e-02   & 2.98e-07 \\
\hline       &   Avar            &  5.46e-07  &  8.19e-12  &  5.46e-07  &  8.19e-12   &  5.46e-07   & 8.19e-12 \\
\hline   1   &  Time (s)            &   32.6730    &        &  26.5839      &       &  4.5110  & \\
\hline       &   Average        &   1.5074    &  0.1000     & 1.5043    &  0.1000     &  1.4809     & 0.1001   \\
\hline       &   Bias           &   7.37e-03  & -2.55e-05   & 4.33e-03  & -1.56e-05   & -1.91e-02   & 1.14e-04 \\
\hline       &   MSE            &   5.75e-05  &  6.98e-10   & 2.40e-05  &  3.40e-10   &  3.29e-02   & 4.15e-07 \\
\hline       &   Avar           &   1.09e-06  &  1.64e-11   & 1.09e-06  &  1.64e-11   &  1.09e-06   & 1.64e-11 \\
\hline
\multicolumn{2}{|c|}{Non-linear Parameters} & $\alpha_2$ & $\beta_2$&  $\alpha_2$ & $\beta_2$& $\alpha_2$ & $\beta_2$ \\
\hline  \multicolumn{2}{|c|}{True values}  & 2.5 & 0.2 &  2.5 & 0.2 & 2.5 & 0.2\\
\hline   0.1 &   Time (s)            & 31.918    &        &  25.5269  &          &  4.5620 & \\
\hline       &   Average         &   2.4999    & 0.2000     &  2.5000    &  0.2000     &  2.4548    &  0.1998  \\
\hline       &   Bias            &  -1.22e-04  & 1.82e-07   & -1.40e-05  & -3.43e-06   & -4.52e-02  & -1.67e-04\\
\hline       &   MSE             &   1.96e-07  & 2.74e-12   &  2.02e-07  &  1.46e-11   &  3.44e-02  &  3.89e-07\\
\hline       &   Avar            &   2.17e-07  & 3.26e-12   &  2.17e-07  &  3.26e-12   &  2.17e-07  &  3.26e-12\\
\hline  0.5  &   Time (s)             & 32.3660   &        &  26.2630  &          & 4.5389 & \\
\hline       &   Average        &   2.5000    &  0.2000     & 2.5001    &  0.2000      &  2.4744    &  0.1999\\
\hline       &   Bias           &  -2.84e-05  & -1.71e-07   & 9.07e-05  & -3.82e-06    & -2.56e-02  & -9.02e-05\\
\hline       &   MSE            &   7.73e-07  &  1.13e-11   & 8.81e-07  &  2.68e-11    &  3.67e-02  &  4.52e-07\\
\hline       &   Avar           &   1.09e-06  &  1.63e-11   & 1.09e-06  &  1.63e-11    &  1.09e-06  &  1.63e-11\\
\hline   1   &  Time (s)            &   32.6730  &      &  26.5839      &       &  4.5110  & \\
\hline       &   Average        &    2.4999     & 0.2000    &  2.5000    &  0.2000      &  2.4707    &  0.1999\\
\hline       &   Bias           &   -1.26e-04   & 1.62e-07  & -2.08e-05  & -3.43e-06    & -2.93e-02  & -8.50e-05\\
\hline       &   MSE            &    1.82e-06   & 2.65e-11  &  2.03e-06  &  4.04e-11    &  2.61e-02  &  3.63e-07\\
\hline       &   Avar           &    2.17e-06   & 3.26e-11  &  2.17e-06  &  3.26e-11    &  2.17e-06  &  3.26e-11\\
\hline
\end{tabular}}
\caption{Estimates of the two component model when sample size is 250}
\label{table:4}
\end{table}

\begin{table}[H]
\centering
\resizebox{0.75\textwidth}{!}{\begin{tabular}{|c|c|c|c|c|c|c|c|}
\hline
 \multicolumn{2}{|c|}{Non-linear Parameters} & $\alpha_1$ & $\beta_1$&  $\alpha_1$ & $\beta_1$& $\alpha_1$ & $\beta_1$ \\
\hline $\sigma^2$ & & \multicolumn{2}{|c|}{LSEs} & \multicolumn{2}{|c|}{ALSEs} & \multicolumn{2}{|c|}{Efficient Algorithm}\\
\hline        \multicolumn{2}{|c|}{True values} & 1.5 & 0.1 &  1.5 & 0.1 & 1.5 & 0.1 \\
\hline   0.1 &   Time (s)          &    61.0879   &               &   55.7359      &                &     8.4870     &             \\
\hline       &   Average       &     1.5020   &   0.1000      &   1.5011       &   0.1000       &    1.4798      &    0.1000    \\
\hline       &   Bias          &     1.98e-03 &  -4.30e-06    &   1.13e-03     &   -2.49e-06    &    -2.02e-02   &    -4.34e-05  \\
\hline       &   MSE           &     4.01e-06 &   1.88e-11    &   1.33e-06     &   6.40e-12     &    1.17e-02    &    2.82e-08  \\
\hline       &   Avar          &     1.37e-08 &   5.12e-14    &   1.37e-08     &   5.12e-14     &    1.37e-08    &    5.12e-14  \\
\hline   0.5 &   Time (s)          &    61.8270   &               &   55.9599      &                &    8.4100      &              \\
\hline       &    Average      &     1.5020   &   0.1000      &   1.5011       &   0.1000       &    1.4840      &    0.1000    \\
\hline       &    Bias         &     1.97e-03 &  -4.29e-06    &   1.13e-03     &   -2.48e-06    &    -1.60e-02   &    -3.96e-05 \\
\hline       &    MSE          &     4.12e-06 &   1.92e-11    &   1.49e-06     &   7.02e-12     &    1.19e-02    &    3.37e-08  \\
\hline       &    Avar         &     6.83e-08 &   2.56e-13    &   6.83e-08     &   2.56e-13     &    6.83e-08    &    2.56e-13  \\
\hline   1   &  Time (s)           &    63.3360   &               &   57.1080      &                &    8.7080      &              \\
\hline       &    Average      &     1.5020   &   0.1000      &   1.5011       &   0.1000       &    1.4832      &    0.1000     \\
\hline       &    Bias         &     1.99e-03 &  -4.32e-06    &   1.14e-03     &   -2.53e-06    &    -1.68e-02   &    -2.97e-05 \\
\hline       &    MSE          &     4.35e-06 &   2.03e-11    &   1.76e-06     &   8.22e-12     &    8.34e-03    &    2.37e-08  \\
\hline       &    Avar         &     1.37e-07 &   5.12e-13    &   1.37e-07     &   5.12e-13     &    1.37e-07    &    5.12e-13  \\
\hline
\multicolumn{2}{|c|}{Non-linear Parameters} & $\alpha_2$ & $\beta_2$&  $\alpha_2$ & $\beta_2$& $\alpha_2$ & $\beta_2$ \\
  \hline  \multicolumn{2}{|c|}{True values}  & 2.5 & 0.2 &  2.5 & 0.2 & 2.5 & 0.2\\
\hline   0.1 &   Time (s)          &    61.0879   &               &    55.7359     &                &     8.4870     &               \\
\hline       &     Average     &     2.4999   &   0.2000      &    2.4987      &    0.2000      &     2.4861     &     0.1999      \\
\hline       &     Bias        &     -5.39e-05&   1.35e-08    &   -1.26e-03    &    2.13e-06    &    -1.39e-02   &    -7.72e-05 \\
\hline       &     MSE         &     2.32e-08 &   7.37e-14    &    1.61e-06    &    4.66e-12    &     3.04e-03   &     1.66e-08 \\
\hline       &     Avar        &     2.72e-08 &   1.02e-13    &    2.72e-08    &    1.02e-13    &     2.72e-08   &     1.02e-13 \\
\hline   0.5 &   Time (s)          &    61.8270   &               &    55.9599     &                &     8.4100     &                 \\
\hline       &     Average     &     2.5000   &   0.2000      &    2.4988      &    0.2000      &     2.5017     &     0.2000 \\
\hline       &     Bias        &     -4.13e-05&  -1.44e-08    &   -1.24e-03    &    2.09e-06    &     1.65e-03   &    -4.06e-05 \\
\hline       &     MSE         &     9.53e-08 &   3.53e-13    &    1.67e-06    &    4.89e-12    &     2.89e-03   &     1.28e-08 \\
\hline       &     Avar        &     1.36e-07 &   5.10e-13    &    1.36e-07    &    5.10e-13    &     1.36e-07   &     5.10e-13 \\
\hline   1   &  Time (s)           &    63.3360   &               &    57.1080     &                &     8.7080     &                   \\
\hline       &    Average      &     2.5000   &   0.2000      &    2.4988      &    0.2000      &     2.5102     &     0.2000  \\
\hline       &    Bias         &     -3.38e-05&  -1.78e-08    &    -1.23e-03   &    2.08e-06    &     1.02e-02   &    -1.78e-05 \\
\hline       &    MSE          &     1.98e-07 &   7.29e-13    &    1.78e-06    &    5.37e-12    &     4.49e-03   &     1.55e-08  \\
\hline       &    Avar         &     2.72e-07 &   1.02e-12    &    2.72e-07    &    1.02e-12    &     2.72e-07   &     1.02e-12  \\
\hline
\end{tabular}}
\caption{Estimates of the two component model when sample size is 500}
\label{table:5}
\end{table}
\justify
This process of data generation and estimation of the unknown parameters is replicated 1000 times and we calculate the average values, bias and MSEs of these estimates. We also report the time taken for the entire simulation process by each of the estimation methods. We compute the asymptotic variance of the estimates to compare the MSEs with them. Simulation results provided in tables~\ref{table:4},~\ref{table:5} and~\ref{table:6}, for the two component model, show that the MSEs of the proposed sequential estimators are well matched to the MSEs of LSEs and they become close as $n$ increases. Also they are comparable to the asymptotic variance of the LSEs. In many cases, it is observed that the MSEs of the ALSEs of the first component, are smaller than the corresponding LSEs. In all the tables, it is consistently observed that compared to the LSEs, computation of the ALSEs takes lesser time.\\

\begin{table}[H]
\centering
\resizebox{0.75\textwidth}{!}{\begin{tabular}{|c|c|c|c|c|c|c|c|}
\hline
 \multicolumn{2}{|c|}{Non-linear Parameters} & $\alpha_1$ & $\beta_1$&  $\alpha_1$ & $\beta_1$& $\alpha_1$ & $\beta_1$ \\
\hline $\sigma^2$ & & \multicolumn{2}{|c|}{LSEs} & \multicolumn{2}{|c|}{ALSEs} & \multicolumn{2}{|c|}{Efficient Algorithm}\\
\hline        \multicolumn{2}{|c|}{True values} & 1.5 & 0.1 &  1.5 & 0.1 & 1.5 & 0.1 \\
 \hline 0.1&   Time (s)             &   124.913  &            &  114.535  &          &   16.7209  &      \\
 \hline    &   Average              &   1.4999   &   0.1000   &  1.5002   &  0.1000  &   1.5206   &   0.1000    \\
 \hline    &   Bias                 &   -7.22e-05&  -1.91e-09 &  1.87e-04 & -3.52e-07&   2.06e-02 &   1.68e-05  \\
 \hline    &   MSE                  &   1.02e-08 &   5.16e-15 &  4.17e-08 &  1.30e-13&   7.69e-03 &   5.54e-09  \\
 \hline    &   Avar                 &   1.71e-09 &   1.60e-15 &  1.71e-09 &  1.60e-15&   1.71e-09 &   1.60e-15  \\
\hline 0.5&   Time (s)             &      118.263     &                 &      115.557     &                 &      16.3389   &                  \\
\hline    &   Average          &      1.4999      &     0.1000      &      1.5002      &     0.1000      &      1.5168    &    0.1000  \\
\hline    &   Bias             &      -7.51e-05   &     1.25e-09    &      1.97e-04    &    -3.61e-07    &      1.68e-02  &    1.20e-05\\
\hline    &   MSE              &      2.91e-08    &     2.47e-14    &      6.41e-08    &     1.54e-13    &      2.32e-03  &    1.73e-09\\
\hline    &   Avar             &      8.53e-09    &     8.00e-15    &      8.53e-09    &     8.00e-15    &      8.53e-09  &    8.00e-15\\
\hline 1  &   Time (s)             &      118.7809    &                 &      114.4330    &                 &      16.3170   &                     \\
\hline    &   Average          &      1.4999      &     0.1000      &      1.5002      &     0.1000      &      1.5111    &    0.1000  \\
\hline    &   Bias             &      -7.32e-05   &     7.49e-10    &      2.04e-04    &    -3.66e-07    &      1.11e-02  &    6.34e-06\\
\hline    &   MSE              &      5.45e-08    &     5.07e-14    &      9.13e-08    &     1.81e-13    &      1.18e-03  &    8.62e-10\\
\hline    &   Avar             &      1.71e-08    &     1.60e-14    &      1.71e-08    &     1.60e-14    &      1.71e-08  &    1.60e-14\\
\hline
\multicolumn{2}{|c|}{Non-linear Parameters} & $\alpha_2$ & $\beta_2$&  $\alpha_2$ & $\beta_2$& $\alpha_2$ & $\beta_2$ \\
    \hline  \multicolumn{2}{|c|}{True values}  & 2.5 & 0.2 &  2.5 & 0.2 & 2.5 & 0.2\\
  \hline 0.1&   Time (s)            &   124.913  &           &   114.535  &          &   16.7209    &    \\
  \hline    &   Average             &   2.5000   &  0.2000   &   2.4998   & 0.2000   &   2.4958     &  0.2000 \\
  \hline    &   Bias                &   2.29e-05 & -1.69e-08 &  -2.44e-04 & 3.25e-07 &  -4.15e-03   & -1.30e-06 \\
  \hline    &   MSE                 &   3.32e-09 &  2.90e-15 &   6.29e-08 & 1.09e-13 &   9.38e-04   &  6.46e-10 \\
  \hline    &   Avar                &   3.40e-09 &  3.19e-15 &   3.40e-09 & 3.19e-15 &   3.40e-09   &  3.19e-15 \\
\hline 0.5&   Time (s)             &      118.263     &                 &      115.557     &                 &      16.3389   &              \\
\hline    &   Average          &      2.5000      &     0.2000      &      2.4998      &     0.2000      &      2.4979    &    0.2000\\
\hline    &   Bias             &      2.08e-05    &     -1.40e-08   &      -2.45e-04   &     3.26e-07    &     -2.13e-03  &    6.78e-07\\
\hline    &   MSE              &      1.43e-08    &     1.31e-14    &      7.47e-08    &     1.20e-13    &      6.51e-04  &    5.28e-10\\
\hline    &   Avar             &      1.70e-08    &     1.59e-14    &      1.70e-08    &     1.59e-14    &      1.70e-08  &    1.59e-14\\
\hline  1 &   Time (s)             &      118.7809    &                 &      114.4330    &                 &      16.3170   &             \\
\hline    &   Average          &      2.5000      &     0.2000      &      2.4997      &     0.2000      &      2.5007    &    0.2000\\
\hline    &   Bias             &      1.01e-05    &     -3.41e-09   &      -2.55e-04   &     3.37e-07    &      7.45e-04  &    2.89e-06\\
\hline    &   MSE              &      2.81e-08    &     2.65e-14    &      9.42e-08    &     1.41e-13    &      1.43e-03  &    8.21e-10\\
\hline    &   Avar             &      3.40e-08    &     3.19e-14    &      3.40e-08    &     3.19e-14    &      3.40e-08  &    3.19e-14\\
\hline
\end{tabular}}
\caption{Estimates of the two component model when sample size is 1000}
\label{table:6}
\end{table}
It is observed that the estimates of the unknown non-linear parameters of the second component for the two component model, that is of $\alpha_2^0$, $\beta_2^0$, have very small bias as compared to those obtained at the first stage, that is of $\alpha_1^0$, $\beta_1^0$ or those obtained for the one component model, $\alpha^0$, $\beta^0$.
Since the proposed ALSEs have desirable properties, it is a good idea to obtain the initial estimates by maximising the periodogram-like function $I(\alpha,\beta)$ as defined in~(\ref{eq:4}) and then carry out the least squares estimation.

\section{Real Data Analysis}\label{section:5}
For illustration, we perform analysis of two speech signal data sets  "AHH" and "AAA". These data have been obtained from a sound instrument at the Speech Signal Processing laboratory of the Indian Institute of Technology Kanpur. We have 469 data points in the "AHH" signal data set and 477 data points in the "AAA" signal data set, both sampled at 10 kHz frequency. Figure~{\ref{fig:AHH}} gives the plot of the observed signal "AHH" and  Figure~{\ref{fig:AAA}} gives the plot of the observed signal "AAA". 
\begin{figure}[H]
\begin{center}
\includegraphics[scale=0.15]{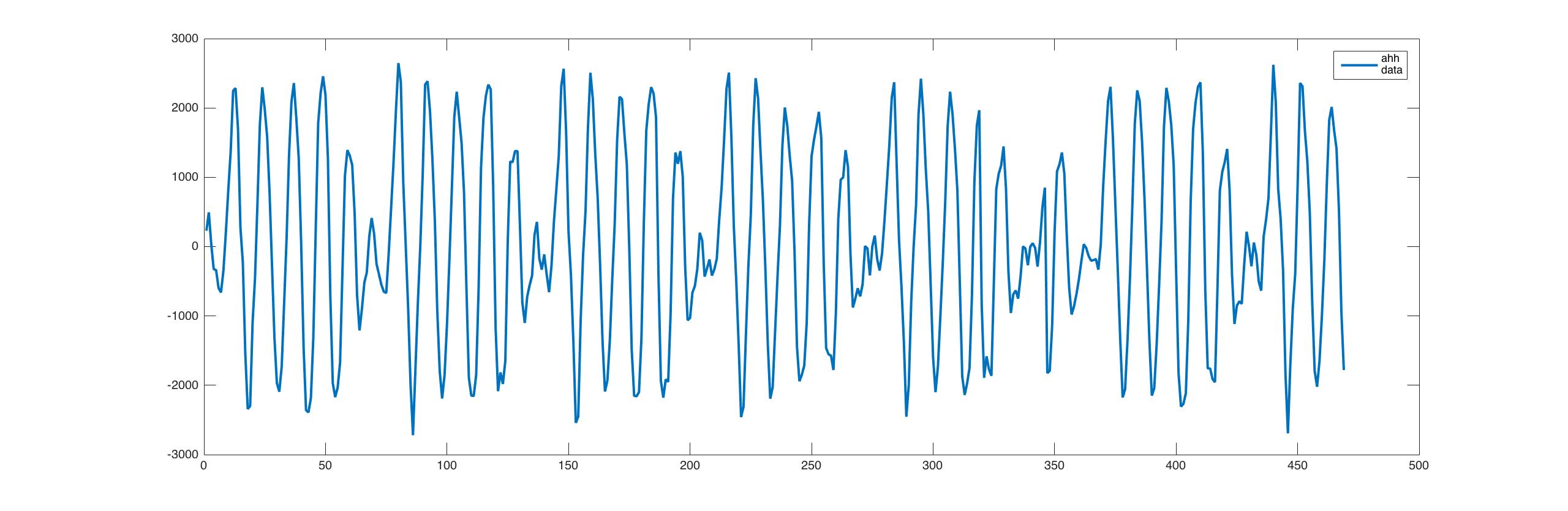}
\caption{AHH: original signal}
\label{fig:AHH}
\end{center}
\end{figure}

\begin{figure}[H]
\begin{center}
\includegraphics[scale=0.15]{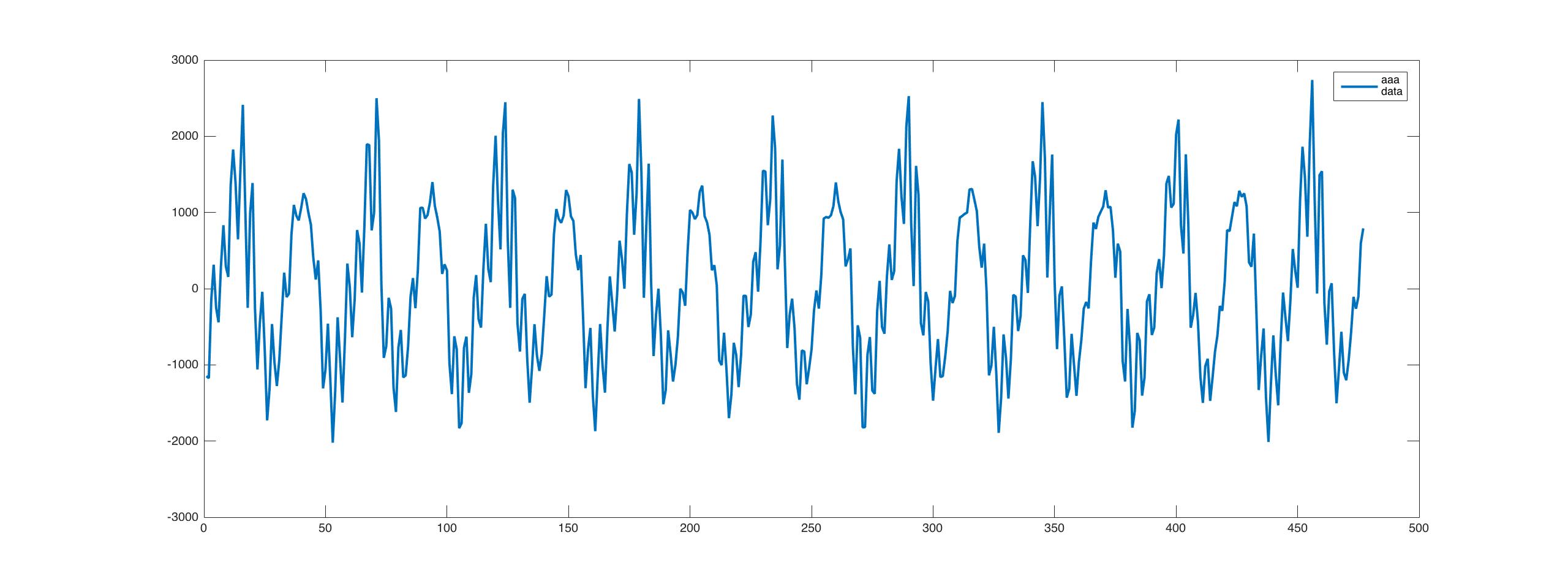}
\caption{AAA: original signal}
\label{fig:AAA}
\end{center}
\end{figure}

\justify
We try to fit a multiple component chirp model to both the data sets, using the proposed sequential estimation procedure which computes ALSEs at each stage. At the same time, we compute the sequential LSEs as proposed by Lahiri et al., \cite{2015} for comparison purposes. To find the initial values of the frequency and frequency rate, at each stage we maximize the periodogram-like function, $I(\alpha,\beta)$ over a fine grid:  $\displaystyle{\bigg(\frac{\pi j}{n}, \frac{\pi k}{n^2}\bigg)}$, $j$ = 1, 2, $\cdots$, $n$, $k$ = 1, 2, $\cdots$, $n^2$.
\justify
For the estimation of the number of components, we use the following form of BIC: 
$$\textmd{BIC}(k) = n \ ln(\textmd{SSE}(k)) + 2\ (4k + 1)\ ln(n). $$
The model order is estimated as the value of $k \in \{1, 2, \cdots, K\}$ for which the BIC is minimum. For the "AHH" data, when we estimate the parameters using sequential least squares estimation procedure, it is evident from Figure~\ref{fig:BIC_LSE} that the number of components that fits this data is 8. Using the proposed sequential ALSEs to fit the model also gives the same estimated number of components which can be seen in Figure~\ref{fig:BIC_ALSE}. The number of components when we estimate the parameters  of the "AAA" data, using sequential least squares estimation procedure is 9, as can be seen from Figure~\ref{fig:BIC_LSE_AAA}. The proposed sequential ALSEs also give the same estimated number of components which can be seen in Figure~\ref{fig:BIC_ALSE_AAA}.
\begin{figure}[H]
\begin{center}
\includegraphics[scale=0.15]{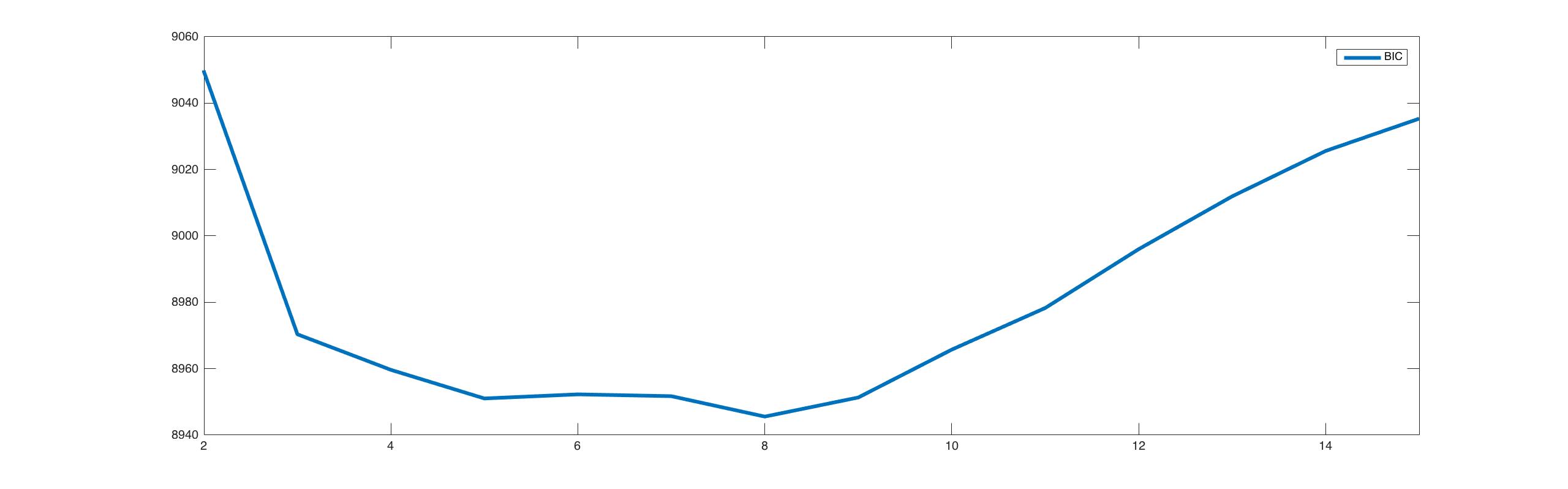}
\caption{BIC plot:  "AHH" data set when estimates are obtained by sequential LSE procedure.}
\label{fig:BIC_LSE}
\end{center}
\end{figure}

\begin{figure}[H]
\begin{center}
\includegraphics[scale=0.15]{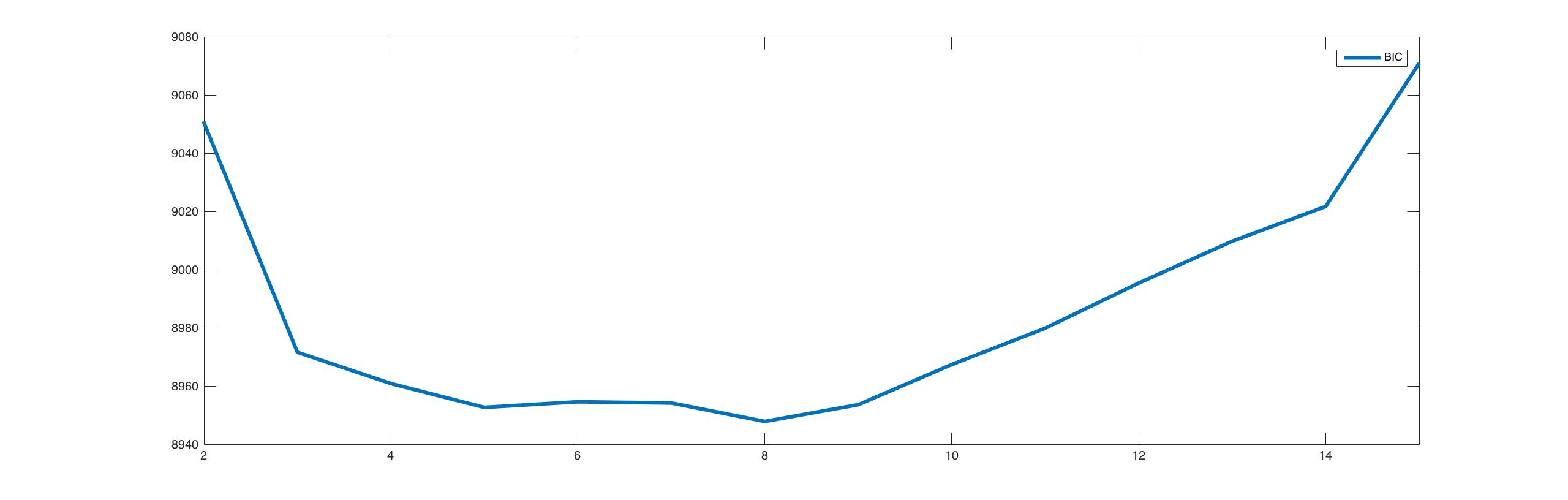}
\caption{BIC plot: "AHH" data set when estimates are obtained by sequential ALSE procedure.}
\label{fig:BIC_ALSE}
\end{center}
\end{figure}

\begin{figure}[H]
\begin{center}
\includegraphics[scale=0.15]{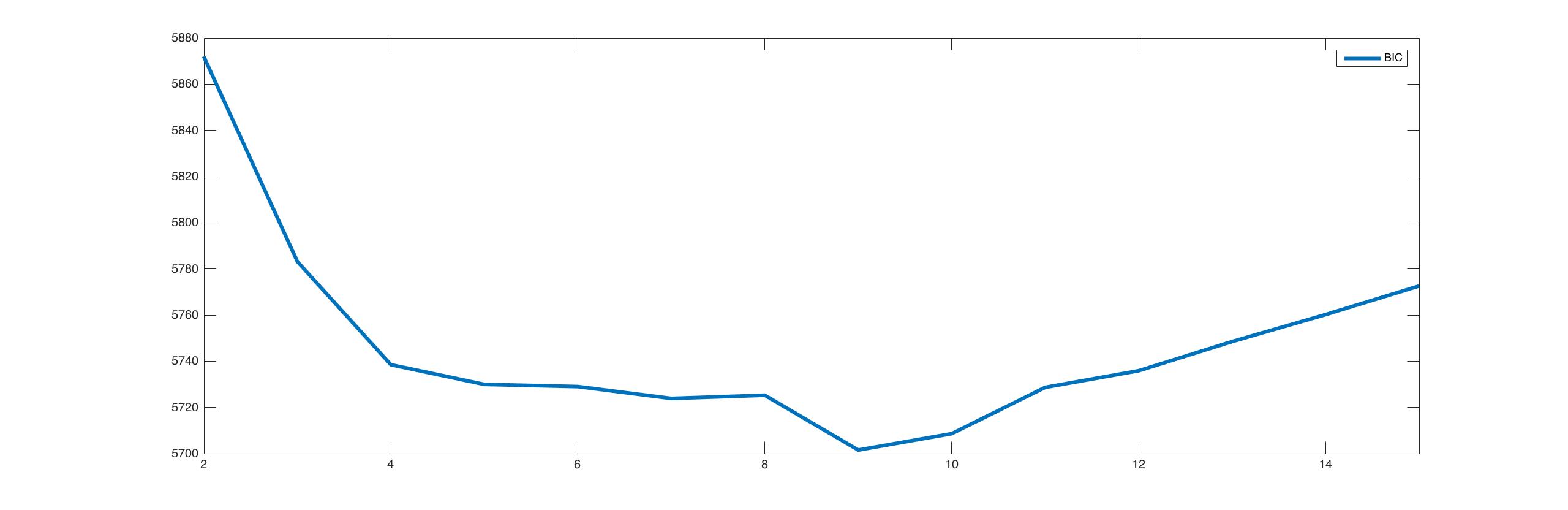}
\caption{BIC plot: "AAA" data set when estimates are obtained by sequential LSE procedure.}
\label{fig:BIC_LSE_AAA}
\end{center}
\end{figure}

\begin{figure}[H]
\begin{center}
\includegraphics[scale=0.15]{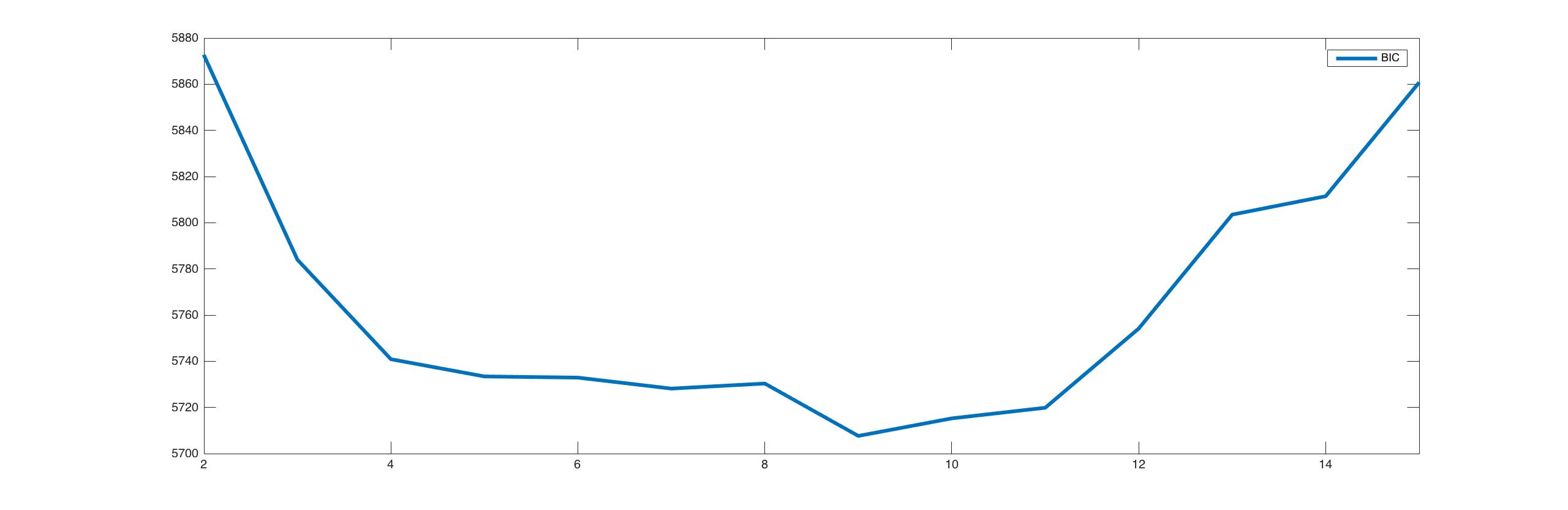}
\caption{BIC plot: "AAA" data set when estimates are obtained by sequential ALSE procedure.}
\label{fig:BIC_ALSE_AAA}
\end{center}
\end{figure}

\justify
Figure~\ref{fig:fit_LSE}  and Figure~\ref{fig:fit_ALSE} gives the observed as well as the fitted signal for the "AHH" data, estimated using the sequential LSEs and using the sequential ALSEs, respectively. We observe from these plots that both the fits look similar. Hence we may conclude from here as well, that the ALSEs are equivalent to the LSEs. Figure~\ref{fig:fit_LSE_aaa}  and Figure~\ref{fig:fit_ALSE_aaa} give the observed as well as the fitted signal for the "AAA" data, estimated using the sequential LSEs and using the sequential ALSEs, respectively.
\begin{figure}[H]
\begin{center}
\includegraphics[scale=0.15]{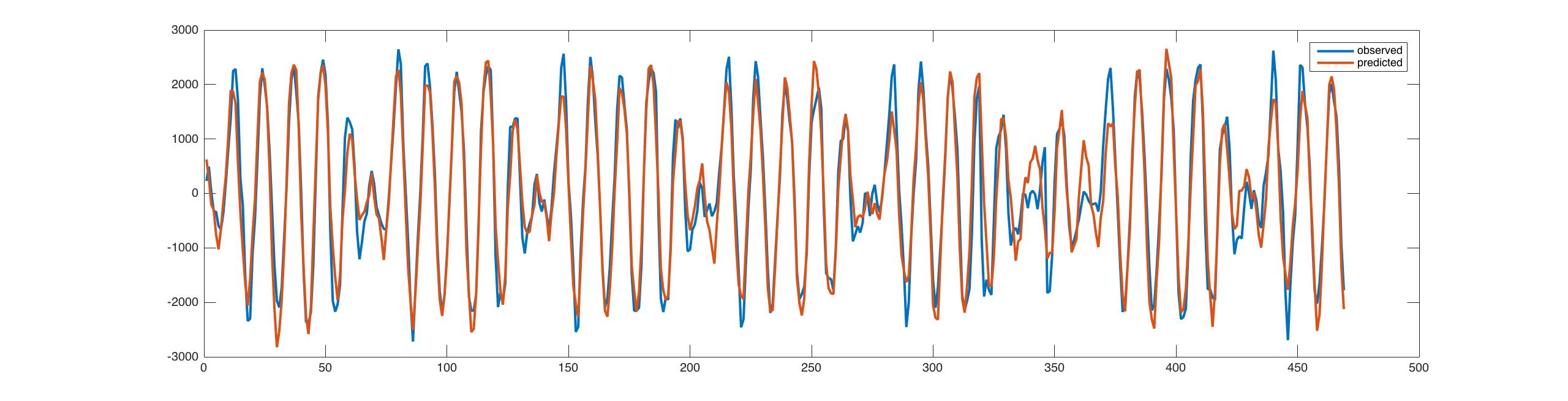}
\caption{Observed "AHH" signal and signal fitted using sequential LSEs. }
\label{fig:fit_LSE}
\end{center}
\end{figure}
\begin{figure}[H]
\begin{center}
\includegraphics[scale=0.15]{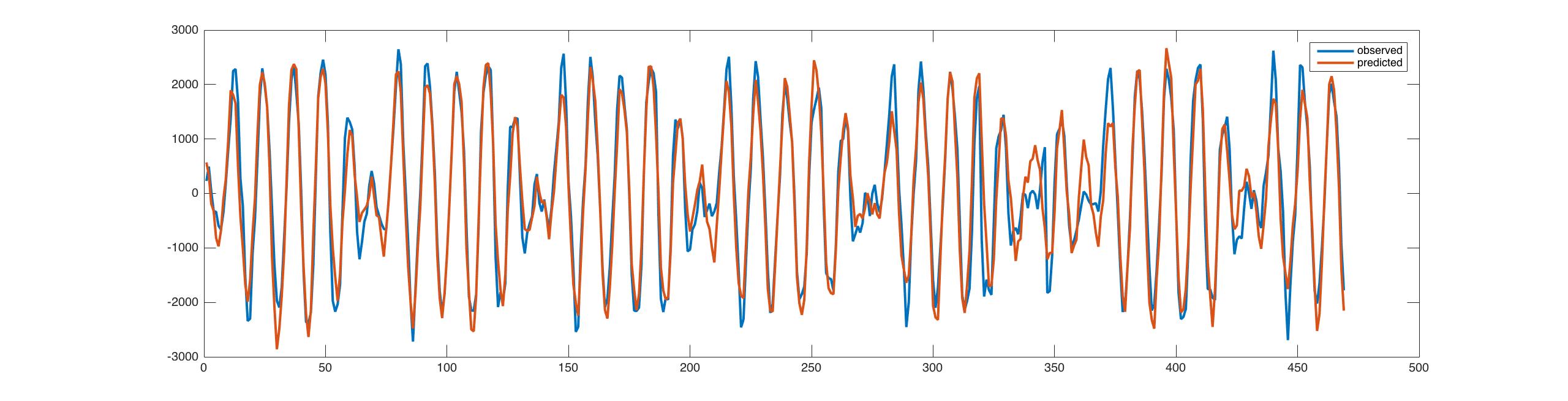}
\caption{Observed "AHH" signal and signal fitted using sequential ALSEs. }
\label{fig:fit_ALSE}
\end{center}
\end{figure}
\begin{figure}[H]
\begin{center}
\includegraphics[scale=0.15]{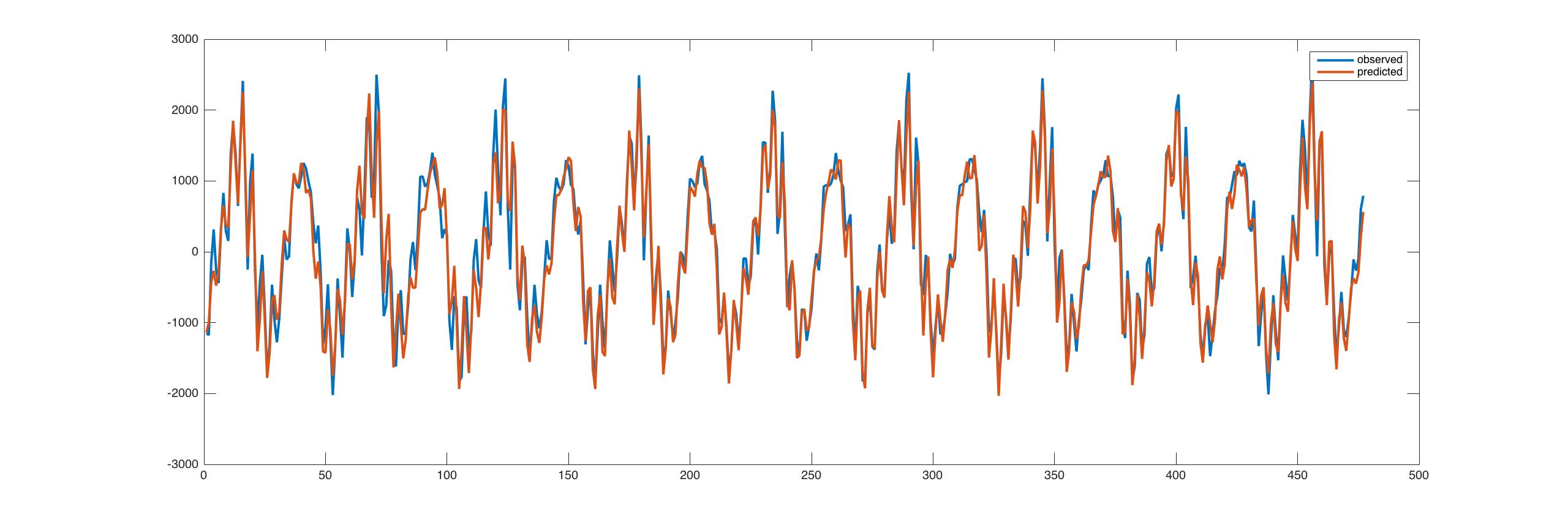}
\caption{Observed "AAA" signal and signal fitted using sequential LSEs. }
\label{fig:fit_LSE_aaa}
\end{center}
\end{figure}
\begin{figure}[H]
\begin{center}
\includegraphics[scale=0.15]{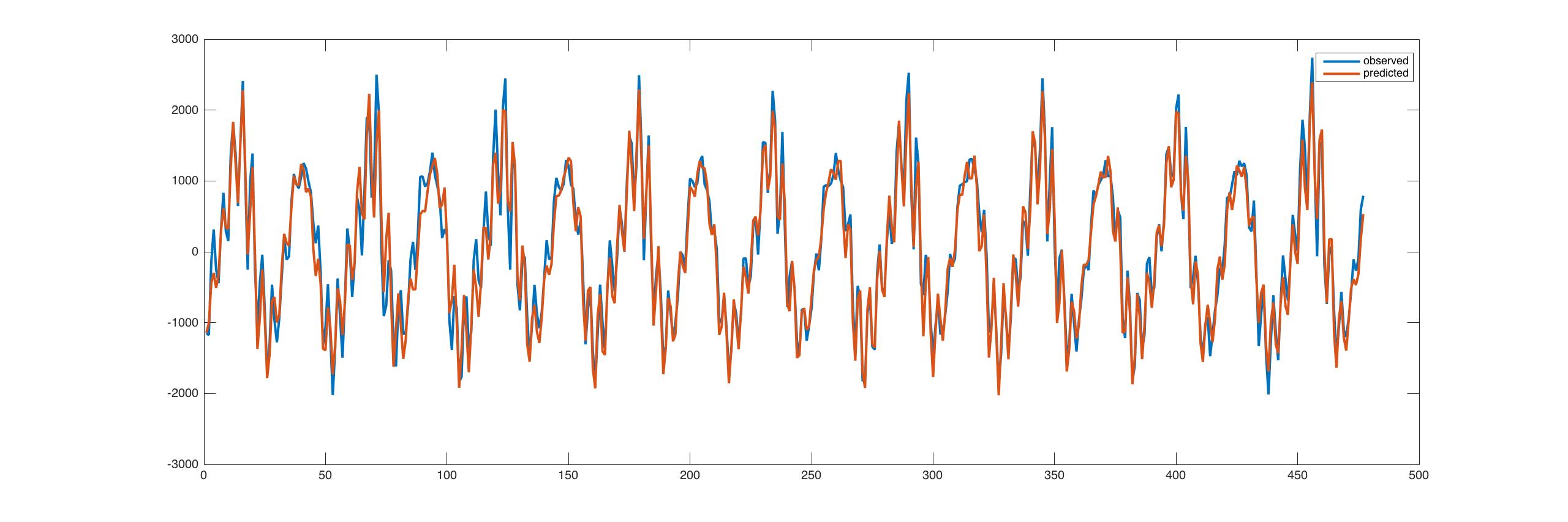}
\caption{Observed "AAA" signal and signal fitted using sequential ALSEs. }
\label{fig:fit_ALSE_aaa}
\end{center}
\end{figure}
\justify
We analyze the residuals by performing an augmented Dickey-Fuller (ADF) test and Kwiatkowski-Phillips-Schmidt-Shin (KPSS) test to check for their stationarity. This is done using in-built R-functions "adf.test" and "kpss.test" in "tseries" package in R. ADF test, tests the null hypothesis of unit-root being present in the time series against the alternative of no unit root, that is, stationarity and KPSS test is used for testing a null hypothesis that an observable time series is stationary around a deterministic trend against the alternative of a unit root. For the "AHH" data set, in the ADF test, we reject the null hypothesis and in KPSS test we do not reject the null hypothesis, and thereby from results of both the tests, we conclude that the residuals are stationary.  For the "AAA" data set, in the ADF test, we reject the null hypothesis and in KPSS test we do not reject the null hypothesis, and thereby from results of both the tests, we conclude that the residuals are stationary. 
Figure~\ref{fig:residual_LSE} -~\ref{fig:residual_ALSE_aaa} provide the residual plots for the two data sets under the two sequential procedures.
\begin{figure}[H]
\begin{center}
\includegraphics[scale=0.25]{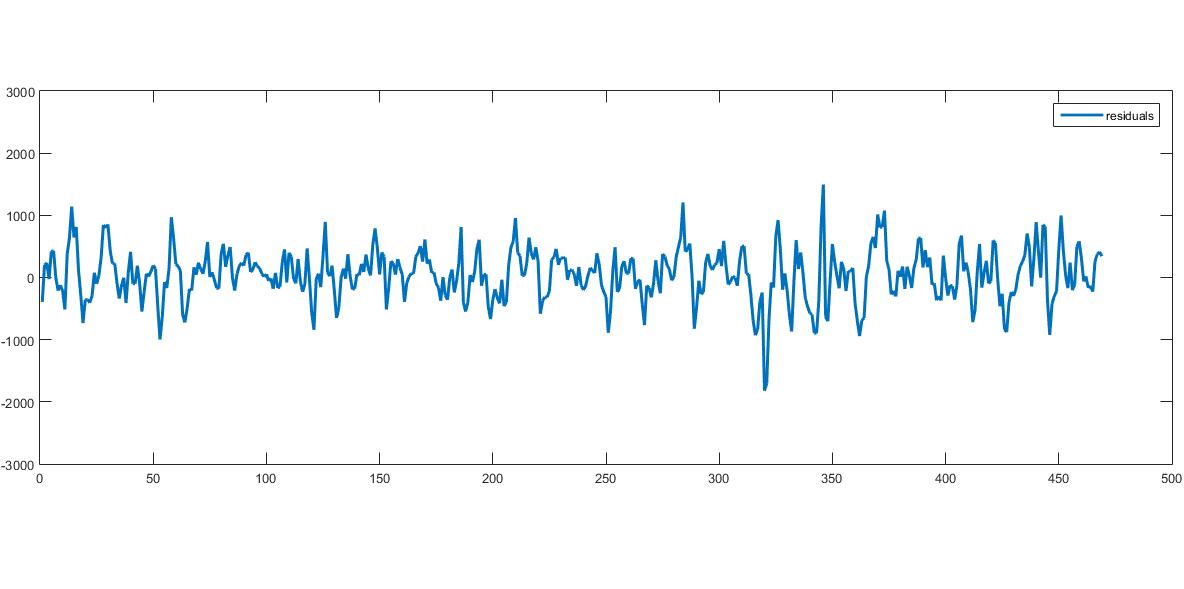}
\caption{Residual plot: of "AHH" data when the estimation is using LSEs.}
\label{fig:residual_LSE}
\end{center}
\end{figure}
\begin{figure}[H]
\begin{center}
\includegraphics[scale=0.25]{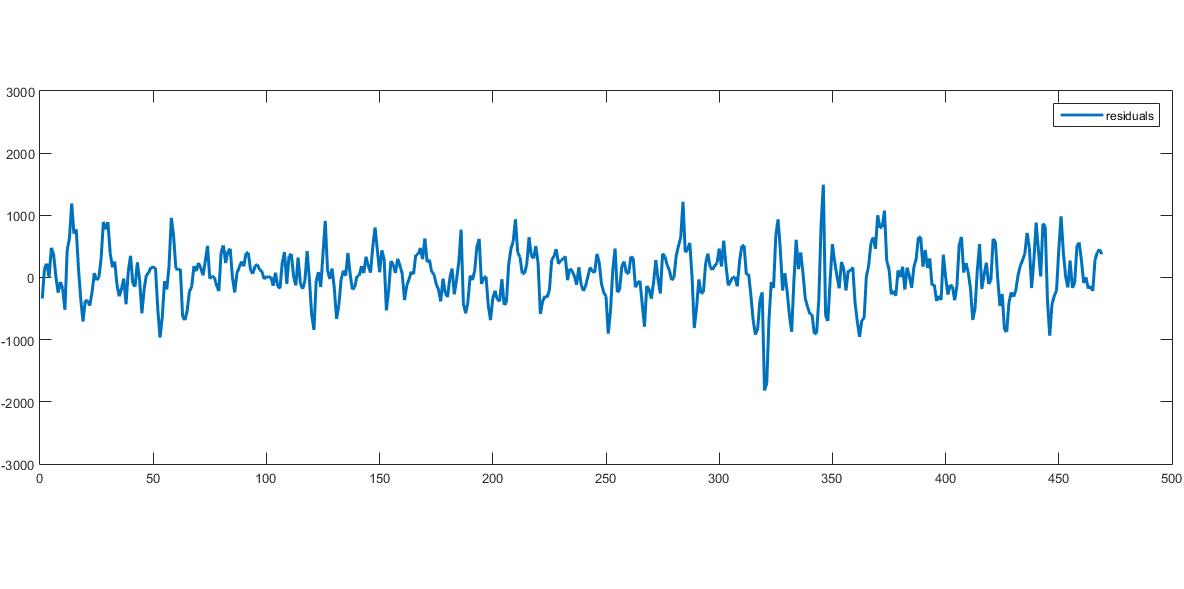}
\caption{Residual plot: of "AHH" data when the estimation is using ALSEs. }
\label{fig:residual_ALSE}
\end{center}
\end{figure}

\begin{figure}[H]
\begin{center}
\includegraphics[scale=0.25]{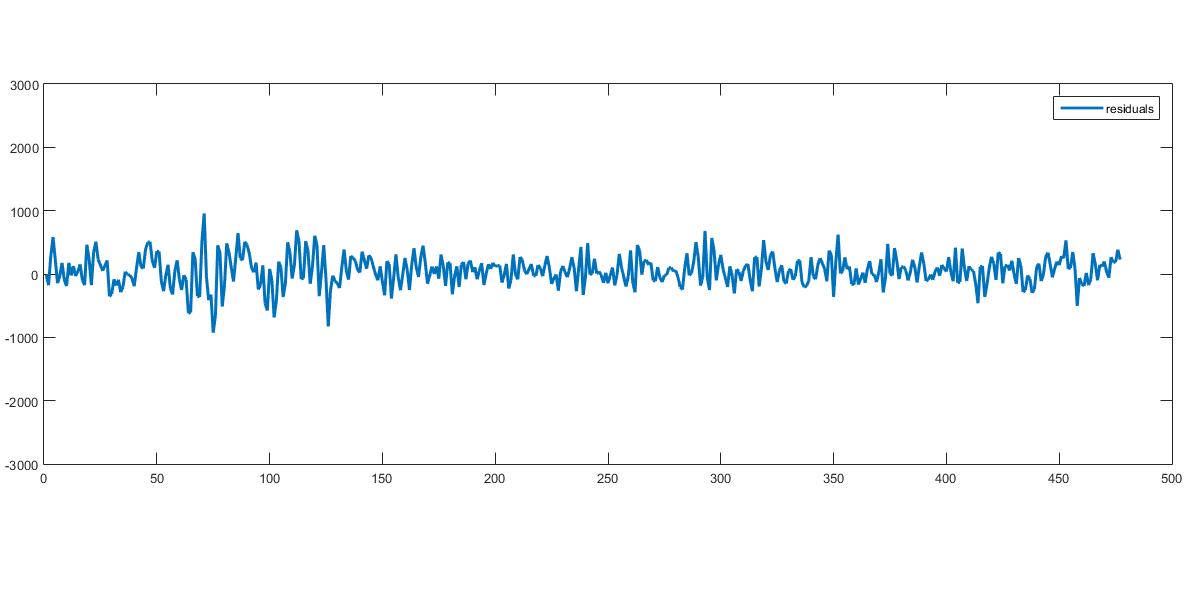}
\caption{Residual plot: of "AAA" data when the estimation is using LSEs. }
\label{fig:residual_LSE_aaa}
\end{center}
\end{figure}
\begin{figure}[H]
\begin{center}
\includegraphics[scale=0.25]{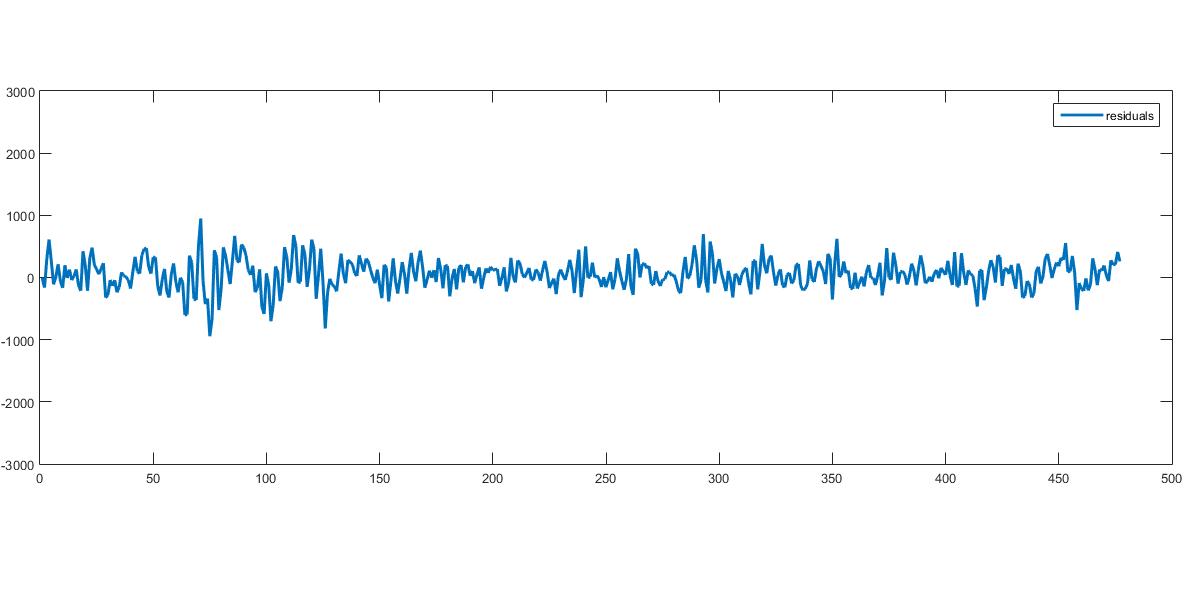}
\caption{Residual plot: of "AAA" data when the estimation is using ALSEs. }
\label{fig:residual_ALSE_aaa}
\end{center}
\end{figure}

\section{Conclusion}\label{section:6}
In this paper, we proposed periodogram-type estimators, called the approximate least squares estimators (ALSEs), for the parameters of a one-dimensional one component chirp model and studied their asymptotic properties. We showed that they are consistent and asymptotically equivalent to the LSEs. Also we obtained the consistency of the ALSEs under weaker conditions than those required for the LSEs. For the multiple component chirp model, we proposed a sequential procedure based on calculating ALSEs and at each step of the sequential procedure establish that these are strongly consistent and asymptotically equivalent to the corresponding sequential  LSEs, having the same rates of convergence. Simulation studies presented in the paper also confirm this large sample equivalence. Hence one may use the periodogram-like estimators as the initial values to find the LSEs. We also perform analysis of two speech signal data sets for illustrative purposes and the performances are quite satisfactory.
\section*{Acknowledgement}
The authors would like to thank two unknown reviewers and the associate editor for their constructive comments which have helped to improve the manuscript significantly.
\begin{appendices}
\section*{Appendix A}\label{appendix:A}
The following lemmas are required to prove Theorem~\ref{theorem:1}.
\begin{lemma}\label{lemma:1}
If \{X(t)\} satisfies Assumption 1, then: \\
\begin{enumerate}[label=(\alph*)]
\item $\sup\limits_{\alpha, \beta}\Bigg|\frac{1}{n^{k+1}}\sum\limits_{t=1}^{n} t^k X(t)\cos(\alpha t + \beta t^2)\Bigg| \rightarrow 0 \ a.s.\ as \ n \rightarrow \infty $,\\ 
\item $\sup\limits_{\alpha, \beta}\Bigg|\frac{1}{n^{k+1}}\sum\limits_{t=1}^{n} t^k X(t)\sin(\alpha t + \beta t^2)\Bigg| \rightarrow 0 \ a.s.\ as \ n \rightarrow \infty $.
\end{enumerate} 
\end{lemma}
\begin{proof}
Refer to Kundu and Nandi \cite{2008}. \\
\end{proof}
\begin{lemma}\label{lemma:2}
If $(\theta_1 , \theta_2)$ $\in$ $(0,\pi) \times (0,\pi)$, then except for a countable number of points, the following results are true:
\begin{enumerate}[label=(\alph*)]
\item $\lim\limits_{n \rightarrow \infty} \frac{1}{n} \sum\limits_{t=1}^{n} \cos(\theta_1 t + \theta_2 t^2) = \lim\limits_{n \rightarrow \infty} \frac{1}{n} \sum\limits_{t=1}^{n} \sin(\theta_1 t + \theta_2 t^2) = 0 $,\\
\item $\lim\limits_{n \rightarrow \infty} \frac{1}{n^{k+1}} \sum\limits_{t=1}^{n} t^k \cos^2(\theta_1 t + \theta_2 t^2) = \frac{1}{2(k + 1)}$, \\
\item $\lim\limits_{n \rightarrow \infty} \frac{1}{n^{k+1}} \sum\limits_{t=1}^{n} t^k \sin^2(\theta_1 t + \theta_2 t^2) = \frac{1}{2(k+1)}$, \\
\item $\lim\limits_{n \rightarrow \infty} \frac{1}{n^{k+1}} \sum\limits_{t=1}^{n} t^k \cos(\theta_1 t + \theta_2 t^2)\sin(\theta_1 t + \theta_2 t^2) = 0 $, \\
\item $\lim\limits_{n \rightarrow \infty}\frac{1}{n}\sum\limits_{t=1}^{n} \cos(\theta_1 t^2) = \lim\limits_{n \rightarrow \infty}\frac{1}{n}\sum\limits_{t=1}^{n} \sin(\theta_1 t^2) = 0$.
\end{enumerate}
for all k = 0, 1, 2, $\cdots$
\end{lemma}
\begin{proof}
Refer to Lahiri et al., \cite{2015}.\\
\end{proof}
\begin{lemma}\label{lemma:3}
Suppose $\tilde{\alpha}$ and $\tilde{\beta}$ are the ALSEs of $\alpha^0$ and $\beta^0$, respectively. Let us denote $\boldsymbol{\xi^0} = (\alpha^0, \beta^0)$ and $$S_c = \{ \boldsymbol{\xi}; \boldsymbol{\xi} = (\alpha , \beta),\ |\boldsymbol{\xi} - \boldsymbol{\xi^0}| > c \}, $$ 
If there exists a $c$ $>$ 0, such that \\
\begin{equation}\label{eq:13}
\limsup\sup_{S_c}\frac{1}{n}\left[I(\boldsymbol{\xi}) - I(\boldsymbol{\xi^0})\right] < \ 0 \ a.s., \\
\end{equation}
then $(\tilde{\alpha},\tilde{\beta}) \ converges \  to \ (\alpha^0,\beta^0)$ almost surely. Here $I(\alpha, \beta)$ is as defined in~(\ref{eq:4}).
\end{lemma}
\begin{proof}
Let us denote $\boldsymbol{\tilde{\xi}}$ by $\boldsymbol{\tilde{\xi}_n}$ = $(\tilde{\alpha}_n,\tilde{\beta}_n)$ and $I(\alpha, \beta)$ by $I_n(\alpha,\beta)$ to emphasize that they depend on $n$. Suppose~(\ref{eq:13}) is true and $\boldsymbol{\tilde{\xi}_n}$ does not converge to  $\boldsymbol{\xi^0}$ as $n$  $\rightarrow \infty.$ Then there exists a $c$ $>$ 0 such that
$$P(\limsup_{n \rightarrow \infty} |\boldsymbol{\tilde{\xi}_n} - \boldsymbol{\xi^0}| > c) \geqslant 0.$$
Hence, $\exists$ a $c$ $>$ 0 and a subsequence $\{\boldsymbol{\tilde{\xi}_{n_k}}\}$ of $\{\boldsymbol{\tilde{\xi}_n}\}$ such that
$ |\boldsymbol{\tilde{\xi}_{n_k}} - \boldsymbol{\xi^0}| $ $>$ $c$ for all $k$ = 1, 2, $\cdots$, that is $\boldsymbol{\tilde{\xi}_{n_k}}$ $\in$ $S_c$ for all $k$ = 1, 2, $\cdots$. Since $\boldsymbol{\tilde{\xi}_{n_k}}$ is the ALSE of $\boldsymbol{\xi^0}$ when $n = n_k$, it maximises $I_{n_k}(\boldsymbol{\xi})$,
\setlength{\belowdisplayskip}{0pt} \setlength{\belowdisplayshortskip}{0pt}
\begin{equation*}
\Rightarrow I_{n_k}(\boldsymbol{\tilde{\xi}_{n_k}}) \geqslant I_{n_k}(\boldsymbol{\xi^0}) \quad
\Rightarrow \frac{1}{n_k}\left[I_{n_k}(\boldsymbol{\tilde{\xi}_{n_k}}) - I_{n_k}(\boldsymbol{\xi^0})\right] \geqslant 0.
\end{equation*}
\begin{flalign*}
&\textmd{Hence, } \limsup\sup_{S_c}\frac{1}{n_k}\left[I_{n_k}(\alpha,\beta) - I_{n_k}(\alpha^0,\beta^0)\right] \geqslant 0.&
\end{flalign*}
Thus, we have $P(\limsup\sup_{S_c}\frac{1}{n_k}\left[I_{n_k}(\alpha,\beta) - I_{n_k}(\alpha^0,\beta^0)\right] \geqslant 0) > 0 $ which contradicts~(\ref{eq:13}). Hence, the result follows.
\end{proof}
\begin{lemma}\label{lemma:4}
Suppose $\tilde{\alpha}$ and $\tilde{\beta}$ are the ALSEs of $\alpha^0$ and $\beta^0$, respectively. Let us define $\boldsymbol{\tilde{\xi}}$ = $(\tilde{\alpha} , \tilde{\beta})$ and $\mathbf{D_1}$ = diag($\frac{1}{n\sqrt{n}},\frac{1}{n^2\sqrt{n}}$), then $$(\boldsymbol{\tilde{\xi}} - \boldsymbol{\xi^0})(\sqrt{n}\mathbf{D_1})^{-1} \xrightarrow{a.s.} 0.$$
\end{lemma}
\begin{proof}
Let us denote $\textbf{I}'(\boldsymbol{\xi})$ as the 1 $\times$ 2 first derivative vector, that is, $\textbf{I}'(\boldsymbol{\xi}) = \large{\begin{pmatrix} \frac{\partial I(\alpha, \beta)}{\partial \alpha} & \frac{\partial I(\alpha, \beta)}{\partial \beta} \end{pmatrix}}$
and $\textbf{I}''(\boldsymbol{\xi})$ as the 2 $\times$ 2 second derivative matrix of $I(\boldsymbol{\xi})$, that is,
$$\textbf{I}''(\boldsymbol{\xi}) = \large{\begin{pmatrix} \frac{\partial^2 I(\alpha, \beta)}{\partial \alpha^2} & \frac{\partial^2 I(\alpha, \beta)}{\partial \alpha \partial \beta} \\ \frac{\partial^2 I(\alpha, \beta)}{\partial \alpha \partial \beta} & \frac{\partial^2 I(\alpha, \beta)}{\partial \beta^2}\end{pmatrix}}.$$
Using multivariate Taylor series expansion of $\textbf{I}'(\boldsymbol{\tilde{\xi}})$ around $\boldsymbol{\xi^0}$, we get:
\begin{equation}\label{eq:14}
\textbf{I}'(\boldsymbol{\tilde{\xi}}) - \textbf{I}'(\boldsymbol{\xi^0}) = (\boldsymbol{\tilde{\xi}} - \boldsymbol{\xi^0})\textbf{I}''(\boldsymbol{\bar{\xi}})
\end{equation}
where $\boldsymbol{\bar{\xi}}$ is such that $|\boldsymbol{\bar{\xi}} - \boldsymbol{\xi^0}| \leqslant |\boldsymbol{\tilde{\xi}} - \boldsymbol{\xi^0}|.$ Since $\textbf{I}'(\boldsymbol{\tilde{\xi}})$ = 0,~(\ref{eq:14}) can be re-written as the following:
\vspace{2pt}$$(\boldsymbol{\tilde{\xi}} - \boldsymbol{\xi^0})(\sqrt{n}\mathbf{D})^{-1} =  \left[-\frac{1}{\sqrt{n}} \textbf{I}'(\boldsymbol{\xi^0})\mathbf{D}\right] \left[\mathbf{D}\textbf{I}''(\boldsymbol{\bar{\xi}})\mathbf{D}\right]^{-1}.$$
Let us first consider,
\begin{flalign*}
&\frac{1}{\sqrt{n}} \textbf{I}'(\boldsymbol{\xi^0})\mathbf{D}  =
\begin{pmatrix}
\frac{\partial I(\alpha^0,\beta^0)}{\partial\alpha}   & \frac{\partial(I(\alpha^0,\beta^0))}{\partial\beta}
\end{pmatrix}
\frac{1}{\sqrt{n}}
\begin{pmatrix}
\frac{1}{n\sqrt{n}} & 0 \\
0                          & \frac{1}{n^2\sqrt{n}}
\end{pmatrix}  = \begin{pmatrix}
\frac{1}{n^2}\frac{\partial I(\alpha^0,\beta^0)}{\partial\alpha} & \frac{1}{n^3}\frac{\partial I(\alpha^0,\beta^0)}{\partial\beta}
\end{pmatrix}.&
\end{flalign*}
\justify 
Using Lemmas~\ref{lemma:1} and~\ref{lemma:2}, it can be shown that:
\begin{flalign*}
\frac{1}{n^2}\frac{\partial I(\alpha^0,\beta^0)}{\partial\alpha} \xrightarrow{a.s.} 0 \textmd{ and } \frac{1}{n^3}\frac{\partial I(\alpha^0,\beta^0)}{\partial\beta} \xrightarrow{a.s.} 0 .
\end{flalign*}
\justify
Thus we have,$\frac{1}{\sqrt{n}} \textbf{I}'(\boldsymbol{\xi^0})\mathbf{D}$ $\xrightarrow{a.s.}$ 0. Now, consider the 2 $\times$ 2 matrix $[\mathbf{D}\textbf{I}''(\boldsymbol{\bar{\xi}})\mathbf{D}]$. Since $\boldsymbol{\tilde{\xi}}$ $\xrightarrow{a.s.}$ $\boldsymbol{\xi^0}$ and  $\textbf{I}''(\boldsymbol{\xi})$ is a continuous function of $\boldsymbol{\xi}$, \\
$$\lim_{n\rightarrow \infty} [\mathbf{D}\textbf{I}''(\boldsymbol{\bar{\xi}})\mathbf{D}] = \lim_{n\rightarrow \infty} [\mathbf{D}\textbf{I}''(\boldsymbol{\xi^0})\mathbf{D}].$$
\justify
\begin{flalign*}
&\mathbf{D}\textbf{I}''(\boldsymbol{\xi^0})\mathbf{D}  =
\begin{pmatrix}
\frac{1}{n\sqrt{n}} & 0 \\
0                          & \frac{1}{n^2\sqrt{n}}
\end{pmatrix}
\begin{pmatrix}
\frac{\partial^2 I(\alpha^0,\beta^0)}{\partial\alpha^2}   & \frac{\partial^2(I(\alpha^0,\beta^0))}{\partial\alpha\partial\beta} \\
\frac{\partial^2(I(\alpha^0,\beta^0))}{\partial\beta\partial\alpha} & \frac{\partial^2 I(\alpha^0,\beta^0)}{\partial\beta^2}
\end{pmatrix}
\begin{pmatrix}
\frac{1}{n\sqrt{n}} & 0 \\
0                          & \frac{1}{n^2\sqrt{n}}
\end{pmatrix}
= \begin{pmatrix}
\frac{1}{n^3}\frac{\partial^2 I(\alpha^0,\beta^0)}{\partial\alpha^2}   & \frac{1}{n^4}\frac{\partial^2(I(\alpha^0,\beta^0))}{\partial\alpha\partial\beta} \\
\frac{1}{n^4}\frac{\partial^2(I(\alpha^0,\beta^0))}{\partial\beta\partial\alpha} & \frac{1}{n^5}\frac{\partial^2 I(\alpha^0,\beta^0)}{\partial\beta^2} 
\end{pmatrix}.&
\end{flalign*}
Again by using Lemmas~\ref{lemma:1} and~\ref{lemma:2} on each element of the above matrix, it can be shown that:
$$\lim_{n\rightarrow \infty} [\mathbf{D}\textbf{I}''(\boldsymbol{\bar{\xi}})\mathbf{D}] = -\textbf{S}, $$
where $\textbf{S}   =
\begin{pmatrix}
\frac{{A^0}^2 + {B^0}^2}{24} & \frac{{A^0}^2 + {B^0}^2}{24} \\
\frac{{A^0}^2 + {B^0}^2}{24} & \frac{2({A^0}^2 + {B^0}^2)}{45}
\end{pmatrix}$
is a positive definite matrix. Hence, \\
$$ (\boldsymbol{\tilde{\xi}} - \boldsymbol{\xi^0})(\sqrt{n}\mathbf{D})^{-1} =  \left[-\frac{1}{\sqrt{n}} \textbf{I}'(\boldsymbol{\xi^0})\mathbf{D}\right] \left[\mathbf{D}\textbf{I}''(\boldsymbol{\bar{\xi}})\mathbf{D}\right]^{-1} \xrightarrow{a.s.} 0.$$
\end{proof}
\justify
Lemma~\ref{lemma:1} and Lemma~\ref{lemma:2} are required to prove Lemma~\ref{lemma:4}. Lemma~\ref{lemma:3} provides a sufficient condition for $\tilde{\alpha}$ and $\tilde{\beta}$ to be strongly consistent. Lemma~\ref{lemma:4} is required to prove strong consistency of the amplitudes $\tilde{A}$ and $\tilde{B}$.  \\
\justify
\emph{Proof of Theorem 1:} To prove the consistency of $\tilde{\alpha}$ and $\tilde{\beta}$, the ALSEs of $\alpha^0$ and $\beta^0$ respectively,
\begin{flalign*}
&\textmd{Consider  } \frac{1}{n}\bigg(I(\alpha,\beta) - I(\alpha^0,\beta^0)\bigg)&\\
& = \frac{1}{n^2}\Bigg[\Bigg\{\sum_{t=1}^{n}y(t)\cos(\alpha t + \beta t^2)\Bigg\}^2 + \Bigg\{\sum_{t=1}^{n}y(t)\sin(\alpha t + \beta t^2)\Bigg\}^2  - \Bigg\{\sum_{t=1}^{n}y(t)\cos(\alpha^0 t + \beta^0 t^2)\Bigg\}^2  & \\
& - \Bigg\{\sum_{t=1}^{n}y(t)\sin(\alpha^0 t + \beta^0 t^2)\Bigg\}^2\Bigg]. &
\end{flalign*}
\justify
Now using Lemmas~\ref{lemma:1} and~\ref{lemma:2}, it can be shown that for some $c$ $>$ 0
\begin{flalign*}
& \limsup\sup_{S_c}\frac{1}{n}\left[I(\alpha,\beta) - I(\alpha^0,\beta^0)\right] = -\lim_{n\rightarrow \infty} \Bigg[\Bigg\{\frac{1}{n}\sum_{t=1}^{n} A^0 \cos^2(\alpha^0 t + \beta^0 t^2)\Bigg\}^2 + \Bigg\{\frac{1}{n}\sum_{t=1}^{n} B^0 \sin^2(\alpha^0 t + \beta^0 t^2)\Bigg\}^2 \Bigg]& \\
& = -\frac{1}{4}({A^0}^2 + {B^0}^2) < 0 \ \ a.s. &
\end{flalign*}
Therefore, $\tilde{\alpha}$ $\xrightarrow{a.s.}$ $\alpha^0$ and $\tilde{\beta}$ $\xrightarrow{a.s.}$ $\beta^0$ by Lemma~\ref{lemma:3}.\\ \qed 
\justify
\emph{Proof of Theorem 2:} 
To prove the consistency of linear parameter estimators $\tilde{A}$ and $\tilde{B}$, observe that
\begin{flalign*}
\tilde{A} = &\frac{2}{n}\sum_{t=1}^{n} y(t)\cos(\tilde{\alpha} t + \tilde{\beta} t^2)
= \frac{2}{n}\sum_{t=1}^{n} \bigg(A^0\cos(\alpha^0 t + \beta^0 t^2) + B^0\sin(\alpha^0 t + \beta^0 t^2) + X(t)\bigg)\cos(\tilde{\alpha} t + \tilde{\beta} t^2). &
\end{flalign*}
Using Lemma~\ref{lemma:1}, $\frac{2}{n}\sum_{t=1}^{n} X(t)\cos(\tilde{\alpha} t + \tilde{\beta} t^2) \rightarrow 0$ a.s. Now using the fact that $\tilde{\alpha} - \alpha^0 = o(\frac{1}{n})$ and $\tilde{\beta} - \beta^0 = o(\frac{1}{n^2})$ (see Lemma~\ref{lemma:4}), expanding $\cos(\tilde{\alpha} t + \tilde{\beta} t^2)$ by multivariate Taylor series around $(\alpha^0 , \beta^0)$ and using trigonometric identities in Lemma~\ref{lemma:2}, we get the desired result. \\ \qed
\begin{comment}
\begin{flalign*}
\tilde{A} & \sim  \frac{2}{n}\sum_{t=1}^{n} \bigg(A^0\cos(\alpha^0 t + \beta^0 t^2) + B^0\sin(\alpha^0 t + \beta^0 t^2)\bigg)\bigg[\cos(\alpha^0 t + \beta^0 t^2) + \begin{pmatrix} 
\frac{\partial \cos(\alpha^0 t + \beta^0 t^2)}{\partial \alpha} & \frac{\partial \cos(\alpha^0 t + \beta^0 t^2)}{\partial \beta}
\end{pmatrix}& \\
&\quad \begin{pmatrix} 
\tilde{\alpha} - \alpha^0 \\ 
\tilde{\beta} - \beta^0
\end{pmatrix}\bigg] &\\
& = \frac{2}{n}\sum_{t=1}^{n} \bigg(A^0\cos(\alpha^0 t + \beta^0 t^2) + B^0\sin(\alpha^0 t + \beta^0 t^2)\bigg)\bigg[\cos(\alpha^0 t + \beta^0 t^2) - t (\tilde{\alpha} - \alpha^0)\sin(\alpha^0 t + \beta^0 t^2) &\\
& \quad \quad \quad  - t ^2 (\tilde{\beta} - \beta^0)\sin(\alpha^0 t + \beta^0 t^2)\bigg] \rightarrow A^0 \  a.s.&
\end{flalign*}
using Lemma~\ref{lemma:4} and the trigonometric identities in Lemma~\ref{lemma:2}. Similarly, one can show that $\tilde{B}$ $\rightarrow$ $B^0$ a.s. \\ \qed
\end{comment}
\section*{Appendix B}\label{appendix:B}
Apart from Lemmas 1-4, we require the following lemma to prove Theorem~\ref{theorem:3}.
\begin{lemma}\label{lemma:5}
If $(\theta_1 , \theta_2)$ $\in$ $(0,\pi) \times (0,\pi)$, then except for a countable number of points, the following results are true:
\begin{enumerate}[label=(\alph*)]
\item $\lim\limits_{n \rightarrow \infty} \frac{1}{\sqrt{n}} \sum_{t=1}^{n} \cos(\theta_1 t + \theta_2 t^2) = \lim\limits_{n \rightarrow \infty} \frac{1}{\sqrt{n}} \sum_{t=1}^{n} \sin(\theta_1 t + \theta_2 t^2) = 0 $,\\
\item $\lim\limits_{n \rightarrow \infty} \frac{1}{n\sqrt{n}} \sum_{t=1}^{n} t \cos(\theta_1 t + \theta_2 t^2) = \lim\limits_{n \rightarrow \infty} \frac{1}{n\sqrt{n}} \sum_{t=1}^{n} t \sin(\theta_1 t + \theta_2 t^2) = 0 $,\\
\item $\lim\limits_{n \rightarrow \infty} \frac{1}{n^2\sqrt{n}} \sum_{t=1}^{n} t^2 \cos(\theta_1 t + \theta_2 t^2) = \lim\limits_{n \rightarrow \infty} \frac{1}{n^2\sqrt{n}} \sum_{t=1}^{n} t^2 \sin(\theta_1 t + \theta_2 t^2) = 0 $.\\
\end{enumerate}
\end{lemma}
\begin{proof}
Consider the exponential sum $\sum\limits_{t=1}^ne^{i(\alpha t + \beta t^2 )}$. By Cauchy Schwartz Inequality, we have:
$$\sum\limits_{t=1}^ne^{i(\alpha t + \beta t^2 )} \leqslant \bigg(\sum\limits_{t=1}^ne^{2i\alpha t}\bigg)^{1/2}\bigg(\sum\limits_{t=1}^ne^{2i\beta t^2}\bigg)^{1/2}.$$
It is easy to show that $\sum\limits_{t=1}^ne^{2i\alpha t}$ is $O(1)$.
Also, using Lemma~\ref{lemma:2} we have:
\begin{flalign*}
\lim_{n \rightarrow \infty}\frac{1}{n}\sum\limits_{t=1}^n \cos(\theta t^2) = \lim_{n \rightarrow \infty}\frac{1}{n}\sum\limits_{t=1}^n \sin(\theta t^2) = 0.
\end{flalign*}
Thus, $\sum\limits_{t=1}^ne^{2i\beta t^2} = o(n)$
$ \Rightarrow \sum\limits_{t=1}^ne^{i(\alpha t + \beta t^2 )} = o(\sqrt{n})$
$ \Rightarrow \lim\limits_{n \rightarrow \infty} \frac{1}{\sqrt{n}} \sum\limits_{t=1}^{n} \cos(\theta_1 t + \theta_2 t^2) = \lim\limits_{n \rightarrow \infty} \frac{1}{\sqrt{n}} \sum\limits_{t=1}^{n} \sin(\theta_1 t + \theta_2 t^2) = 0.$\\ \\
Now let us define $E_1 = \sum\limits_{t=1}^n t e^{i(\alpha t + \beta t^2 )}$ and $E_2 = \sum\limits_{t=1}^n t^2 e^{i(\alpha t + \beta t^2 )}.$
\begin{flalign*}
&\textmd{Consider}\  |E_1|\ = |\sum\limits_{t=1}^n t e^{i(\alpha t + \beta t^2 )}| = |n\sum\limits_{t=1}^n e^{i(\alpha t + \beta t^2 )} - \sum\limits_{t=1}^n (n - t) e^{i(\alpha t + \beta t^2 )}| & \\
&\quad \quad \quad \quad \quad \quad \leqslant |n\sum\limits_{t=1}^n e^{i(\alpha t + \beta t^2 )}| + |\sum\limits_{t=1}^1 e^{i(\alpha t + \beta t^2 )} + \sum\limits_{t=1}^2 e^{i(\alpha t + \beta t^2 )} + \cdots + \sum\limits_{t=1}^{n-1} e^{i(\alpha t + \beta t^2 )}| &\\
&\quad \quad \quad \quad \quad \quad = o(n\sqrt{n}).
\end{flalign*}
\justify
Similarly, it can be shown that $|E_2| = o(n^2\sqrt{n})$. Hence, the result follows.\\ \\
\end{proof}
\justify 
\emph{Proof of Theorem 3:} Let $Q(\boldsymbol{\theta})$ be the error sum of squares, then 
\begin{flalign*}
& \frac{1}{n}Q(\boldsymbol{\theta}) = \frac{1}{n}\sum_{t=1}^{n} \bigg(y(t) - (A\cos(\alpha t + \beta t^2) + B \sin(\alpha t + \beta t^2))\bigg)^2 &\\
& = \frac{1}{n}\sum_{t=1}^{n}y(t)^2 - \frac{2}{n}\sum_{t=1}^{n} y(t)\{A\cos(\alpha t + \beta t^2) + B \sin(\alpha t + \beta t^2)\} + \frac{1}{n}\sum_{t=1}^{n} (A\cos(\alpha t + \beta t^2) + B \sin(\alpha t + \beta t^2))^2 \\
& = \frac{1}{n}\sum_{t=1}^{n}y(t)^2 - \frac{2}{n}\sum_{t=1}^{n} y(t)\{A\cos(\alpha t + \beta t^2) + B \sin(\alpha t + \beta t^2)\} + \frac{1}{2}\bigg(A^2 + B^2\bigg) + o(1)&\\
& = C - \frac {1}{n}J(\boldsymbol{\theta}) + o(1).&\\
&Here, \ C = \frac{1}{n} \sum_{t=1}^{n}y^2(t) \textmd{ and }
\frac {1}{n}J(\boldsymbol{\theta}) = \frac{2}{n}\sum_{t=1}^{n} y(t)\{A\cos(\alpha t + \beta t^2) + B \sin(\alpha t + \beta t^2)\}  - \frac{A^2 + B^2}{2}.&
\end{flalign*}
Now we compute the first derivative of $\frac{1}{n}J(\boldsymbol{\theta})$ and $\frac{1}{n}Q(\boldsymbol{\theta})$  at $\boldsymbol{\theta}$ = $\boldsymbol{\theta^0}$ and using Lemmas~\ref{lemma:1},~\ref{lemma:2} and~\ref{lemma:5}, we obtain the following relation between them:
\begin{flalign*}
&\frac{1}{n}\textbf{Q}'(\boldsymbol{\theta^0}) \mathbf{D} = -\frac{1}{n}\textbf{J}'(\boldsymbol{\theta^0}) \mathbf{D} + 
\begin{pmatrix} 
o(\frac{1}{\sqrt{n}}) \\ \\
o(\frac{1}{\sqrt{n}}) \\ \\
o(\sqrt{n}) \\ \\
o(n\sqrt{n}) \\
\end{pmatrix}^{T} \mathbf{D}.& \\
%&\Rightarrow Q'(\boldsymbol{\theta}) D = -J'(\boldsymbol{\theta}) D + 
%\begin{pmatrix} 
%o(\sqrt{n}) \\ \\
%o(\sqrt{n}) \\ \\
%o(n\sqrt{n}) \\ \\
%o(n^2\sqrt{n}) \\
%\end{pmatrix}^{T} D&\\
&\Rightarrow \lim\limits_{n \rightarrow \infty} \textbf{Q}'(\boldsymbol{\theta^0}) \mathbf{D} =  \lim\limits_{n \rightarrow \infty} - \textbf{J}'(\boldsymbol{\theta^0}) \mathbf{D},\ \textmd{since} \ \lim\limits_{n \rightarrow \infty} \begin{pmatrix} 
o(\sqrt{n}) \\ \\
o(\sqrt{n}) \\ \\
o(n\sqrt{n}) \\ \\
o(n^2\sqrt{n}) \\
\end{pmatrix}^{T} \mathbf{D} = 0.&
\end{flalign*}
\justify
Note that $\tilde{A}$ = $\hat{A}(\alpha , \beta)$ and $\tilde{B}$ = $\hat{B}(\alpha , \beta).$,
therefore substituting $\tilde{A}$, $\tilde{B}$ in $J(\boldsymbol{\theta})$,we have: $$J(\tilde{A} , \tilde{B} , \alpha , \beta) =  I(\alpha , \beta).$$
Hence the estimator of $\boldsymbol{\theta}$ which maximizes $J(\boldsymbol{\theta})$ is equivalent to $\boldsymbol{\tilde{\theta}}$, the ALSE of $\boldsymbol{\theta^0}$. Thus, the ALSE $\boldsymbol{\tilde{\theta}}$ in terms of $J(\boldsymbol{\theta})$ can be written as:\\
$$(\boldsymbol{\tilde{\theta}} - \boldsymbol{\theta^0}) = -\textbf{J}'(\boldsymbol{\theta^0})[\textbf{J}''(\boldsymbol{\bar{\theta}})]^{-1}.$$
$$\Rightarrow (\boldsymbol{\tilde{\theta}} - \boldsymbol{\theta^0})\mathbf{D}^{-1} = -[\textbf{J}'(\boldsymbol{\theta^0})\mathbf{D}][\mathbf{D}\textbf{J}''(\boldsymbol{\bar{\theta}})\mathbf{D}]^{-1}.$$ \\
Now we know that, $\lim\limits_{n \rightarrow \infty}[\mathbf{D}\textbf{J}''(\boldsymbol{\bar{\theta}})\mathbf{D}] = \lim\limits_{n \rightarrow \infty}[\mathbf{D}\textbf{J}''(\boldsymbol{\theta^0})\mathbf{D}]$.
Comparing the corresponding elements of the second derivative matrices $\mathbf{D}\textbf{J}''(\boldsymbol{\theta^0})\mathbf{D}$ and $\mathbf{D}\textbf{Q}''(\boldsymbol{\theta^0})\mathbf{D}$ after using Lemmas~\ref{lemma:1} and~\ref{lemma:2} on each of the derivatives as done for the first derivative vectors above (involves lengthy calculations), we obtain the following relation: $$ \lim_{n \rightarrow \infty} \mathbf{D}\textbf{J}''(\boldsymbol{\theta^0})\mathbf{D} = - \lim_{n \rightarrow \infty}\mathbf{D}\textbf{Q}''(\boldsymbol{\theta^0})\mathbf{D} = -\boldsymbol{\Sigma}^{-1}.$$
where, $$\boldsymbol{\Sigma}^{-1} = \left[\begin{array}{cccc}1 & 0 & \frac{B^0}{2} & \frac{B^0}{3} \\0 & 1 & -\frac{A^0}{2} & -\frac{A^0}{3} \\\frac{B^0}{2} & -\frac{A^0}{2} & \frac{{A^0}^2 + {B^0}^2}{3} & \frac{{A^0}^2 + {B^0}^2}{4} \\ \frac{B^0}{3} & -\frac{A^0}{3} & \frac{{A^0}^2 + {B^0}^2}{4} & \frac{{A^0}^2 + {B^0}^2}{5}\end{array}\right].$$\\
\justify
Thus, we have,
\begin{flalign*}
&(\boldsymbol{\tilde{\theta}} - \boldsymbol{\theta^0})\mathbf{D}^{-1} = -[\textbf{J}'(\boldsymbol{\theta^0})\mathbf{D}][\mathbf{D}\textbf{J}''(\boldsymbol{\bar{\theta}})\mathbf{D}]^{-1}.&\\
&\Rightarrow \lim_{n \rightarrow \infty} (\boldsymbol{\tilde{\theta}} - \boldsymbol{\theta^0})\mathbf{D}^{-1} = -\lim_{n \rightarrow \infty} [\textbf{J}'(\boldsymbol{\theta^0})\mathbf{D}] \lim_{n \rightarrow \infty}[\mathbf{D}\textbf{J}''(\boldsymbol{\bar{\theta}})\mathbf{D}]^{-1}.&\\
&\Rightarrow \lim_{n \rightarrow \infty}(\boldsymbol{\tilde{\theta}} - \boldsymbol{\theta^0})\mathbf{D}^{-1}  = -\lim_{n \rightarrow \infty}[\textbf{Q}'(\boldsymbol{\theta^0})\mathbf{D}]\lim_{n \rightarrow \infty}[\mathbf{D}\textbf{Q}''(\boldsymbol{\bar{\theta}})\mathbf{D}]^{-1}.&
\end{flalign*}
Using the result of Kundu and Nandi \cite{2008}, it follows that the right hand side is equal to  $\lim_{n \rightarrow \infty}(\boldsymbol{\hat{\theta}} - \boldsymbol{\theta^0})\mathbf{D}^{-1}$.
Hence,
$$ \lim_{n \rightarrow \infty}(\boldsymbol{\tilde{\theta}} - \boldsymbol{\theta^0})\mathbf{D}^{-1}  =  \lim_{n \rightarrow \infty}(\boldsymbol{\hat{\theta}} - \boldsymbol{\theta^0})\mathbf{D}^{-1}.$$
\justify
It follows that LSE, $\boldsymbol{\hat{\theta}}$ and ALSE, $\boldsymbol{\tilde{\theta}}$ of $\boldsymbol{\theta^0}$ of model~(\ref{eq:5}) are asymptotically equivalent in distribution.\\
Therefore, asymptotic distribution of $\boldsymbol{\tilde{\theta}}$ is same as that of $\boldsymbol{\hat{\theta}}$. \\ \qed
\section*{Appendix C}\label{appendix:C}
\justify
The following Lemmas are required to prove the Theorem~\ref{theorem:4}: \\
\begin{lemma}\label{lemma:6}
Consider the following set $S_{c} = \{(\alpha,\beta): |\alpha - \alpha_1^0| > c\ or \ |\beta - \beta_1^0| > c\}$ for any c $>$ 0. If for some $c$ $>$0, $$\limsup \sup_{S_c} \frac{1}{n} (I_1(\alpha,\beta) - I_1(\alpha_1^0,\beta_1^0)) < 0 \ \ a.s., $$  then $(\tilde{\alpha_1},\tilde{\beta_1})$ is a strongly consistent estimator of  $(\alpha_1^0,\beta_1^0)$. Here $I_1(\alpha, \beta)$ is as defined in~(\ref{eq:10}).
\end{lemma}
\begin{proof}
This proof can be obtained on the same lines as Lemma~\ref{lemma:3}.\\
\end{proof}
\begin{lemma}\label{lemma:7}
Let  $\boldsymbol{\tilde{\xi_1}}$ = $(\tilde{\alpha_1} , \tilde{\beta_1})$ and $\mathbf{D_1}$ = diag($\frac{1}{n\sqrt{n}},\frac{1}{n^2\sqrt{n}}$), then $$(\boldsymbol{\tilde{\xi_1}} - \boldsymbol{\xi_1^0})(\sqrt{n}\mathbf{D_1})^{-1} \xrightarrow{a.s.} 0.$$
\end{lemma}
\begin{proof}
Let us denote $\mathbf{I_1}'(\boldsymbol{\xi_1})$ as the 1 $\times$ 2 first derivative matrix and $\mathbf{I_1}''(\boldsymbol{\xi_1})$ as the 2 $\times$ 2 second derivative matrix of $I_1(\boldsymbol{\xi_1})$. Now, using multivariate Taylor series expansion of $\mathbf{I_1}'(\boldsymbol{\tilde{\xi_1}})$ around $\boldsymbol{\xi_1^0}$, we get:
\begin{equation}\label{eq:15}
\mathbf{I_1}'(\boldsymbol{\tilde{\xi_1}}) - \mathbf{I_1}'(\boldsymbol{\xi_1^0}) = (\boldsymbol{\tilde{\xi_1}} - \boldsymbol{\xi_1^0})\mathbf{I_1}''(\boldsymbol{\bar{\xi_1}})
\end{equation}
where $\boldsymbol{\bar{\xi_1}}$ is such that $|\boldsymbol{\bar{\xi_1}} - \boldsymbol{\xi_1^0}| \leqslant |\boldsymbol{\tilde{\xi_1}} - \boldsymbol{\xi_1^0}|.$
\justify
Since $\mathbf{I_1}'(\boldsymbol{\tilde{\xi_1}})$ = 0,~(\ref{eq:15}) can be written as
$$(\boldsymbol{\tilde{\xi_1}} - \boldsymbol{\xi_1^0})D^{-1} = \left[-\mathbf{I_1}'(\boldsymbol{\xi_1^0})\mathbf{D}\right] \left[\mathbf{D}\mathbf{I_1}''(\boldsymbol{\bar{\xi_1}})\mathbf{D}\right]^{-1}. $$
\vspace{-\baselineskip}
Dividing by $\sqrt{n}$ the above expression becomes
\vspace{3pt}$$(\boldsymbol{\tilde{\xi_1}} - \boldsymbol{\xi_1^0})(\sqrt{n}\mathbf{D})^{-1} =  \left[-\frac{1}{\sqrt{n}} \mathbf{I_1}'(\boldsymbol{\xi_1^0})\mathbf{D}\right] \left[\mathbf{D}\mathbf{I_1}''(\boldsymbol{\bar{\xi_1}})\mathbf{D}\right]^{-1}.$$
Let us first consider $\frac{1}{\sqrt{n}} \mathbf{I_1}'(\boldsymbol{\xi_1^0})\mathbf{D}.$
\begin{flalign*}
\frac{1}{\sqrt{n}} \mathbf{I_1}'(\boldsymbol{\xi_1^0})\mathbf{D} 
& =
\begin{pmatrix}
\frac{\partial I_1(\alpha_1^0,\beta_1^0)}{\partial\alpha}   & \frac{\partial(I_1(\alpha_1^0,\beta_1^0))}{\partial\beta}
\end{pmatrix}
\frac{1}{\sqrt{n}}
\begin{pmatrix}
\frac{1}{n\sqrt{n}} & 0 \\
0                          & \frac{1}{n^2\sqrt{n}}
\end{pmatrix}&\\
& = \begin{pmatrix}
\frac{1}{n^2}\frac{\partial I_1(\alpha_1^0,\beta_1^0)}{\partial\alpha} & \frac{1}{n^3}\frac{\partial I_1(\alpha_1^0,\beta_1^0)}{\partial\beta}
\end{pmatrix}.&
\end{flalign*}
Using Lemmas~\ref{lemma:1} and~\ref{lemma:2}, one can show that both the elements of the above vector go to zero as n $\rightarrow$ $\infty$, almost surely.\\ \\
Thus $\frac{1}{\sqrt{n}} \mathbf{I_1}'(\boldsymbol{\xi_1^0})\mathbf{D}$ $\xrightarrow{a.s.}$ 0.\\ \\
Consider the 2 $\times$ 2 matrix $[\mathbf{D}\mathbf{I_1}''(\boldsymbol{\bar{\xi_1}})\mathbf{D}]$. Since $\boldsymbol{\tilde{\xi_1}}$ $\xrightarrow{a.s.}$ $\boldsymbol{\xi_1^0}$ and  $\mathbf{I_1}''(\boldsymbol{\xi_1})$ is a continuous function of $\xi_1$,
$$\lim_{n\rightarrow \infty} [\mathbf{D}\mathbf{I_1}''(\boldsymbol{\bar{\xi_1}})\mathbf{D}] = \lim_{n\rightarrow \infty} [\mathbf{D}\mathbf{I_1}''(\boldsymbol{\xi_1^0})\mathbf{D}].$$
\justify
Let us now look at the elements of the matrix  $\mathbf{D}\mathbf{I_1}''(\boldsymbol{\xi_1^0})\mathbf{D}.$
\begin{flalign*}
&\mathbf{D}\mathbf{I_1}''(\boldsymbol{\xi_1^0})\mathbf{D}  =
\begin{pmatrix}
\frac{1}{n\sqrt{n}} & 0 \\
0                          & \frac{1}{n^2\sqrt{n}}
\end{pmatrix}
\begin{pmatrix}
\frac{\partial^2 I_1(\alpha_1^0,\beta_1^0)}{\partial\alpha^2}   & \frac{\partial^2(I_1(\alpha_1^0,\beta_1^0))}{\partial\alpha\partial\beta} \\
\frac{\partial^2(I_1(\alpha_1^0,\beta_1^0))}{\partial\beta\partial\alpha} & \frac{\partial^2 I_1(\alpha_1^0,\beta_1^0)}{\partial\beta^2}
\end{pmatrix}
\begin{pmatrix}
\frac{1}{n\sqrt{n}} & 0 \\
0                          & \frac{1}{n^2\sqrt{n}}
\end{pmatrix} = \begin{pmatrix}
\frac{1}{n^3}\frac{\partial^2 I_1(\alpha_1^0,\beta_1^0)}{\partial\alpha^2}   & \frac{1}{n^4}\frac{\partial^2(I_1(\alpha_1^0,\beta_1^0))}{\partial\alpha\partial\beta} \\
\frac{1}{n^4}\frac{\partial^2(I_1(\alpha_1^0,\beta_1^0))}{\partial\beta\partial\alpha} & \frac{1}{n^5}\frac{\partial^2 I_1(\alpha_1^0,\beta_1^0)}{\partial\beta^2} 
\end{pmatrix}.&
\end{flalign*}
Again using Lemmas~\ref{lemma:1} and~\ref{lemma:2}, we obtain the following:
\begin{flalign*}
&\frac{1}{n^3}\frac{\partial^2 I_1(\alpha_1^0,\beta_1^0)}{\partial\alpha^2}  \xrightarrow{a.s} - \Bigg(\frac{{A_1^0}^2}{24} + \frac{{B_1^0}^2}{24}\Bigg), \quad \frac{1}{n^4}\frac{\partial^2(I_1(\alpha_1^0,\beta_1^0))}{\partial\alpha\partial\beta}  \xrightarrow{a.s} - \Bigg(\frac{{A^0}^2}{24} + \frac{{B^0}^2}{24}\Bigg),\ \textmd{and} &\\
&\frac{1}{n^5}\frac{\partial^2 I_1(\alpha_1^0,\beta_1^0)}{\partial\beta^2}   \xrightarrow{a.s}  - \Bigg(\frac{{2A_1^0}^2}{45} + \frac{{2B_1^0}^2}{45}\Bigg).&
\end{flalign*}
Thus, $\lim_{n\rightarrow \infty} [\mathbf{D}\mathbf{I_1}''(\boldsymbol{\bar{\xi_1}})\mathbf{D}] = -\mathbf{S_1},$
where $\mathbf{S_1} =
\begin{pmatrix}
\frac{{A_1^0}2 + {B_1^0}^2}{24} & \frac{{A_1^0}^2 + {B_1^0}^2}{24} \\
\frac{{A_1^0}^2 + {B_1^0}^2}{24} & \frac{2({A_1^0}^2 + {B_1^0}^2)}{45}
\end{pmatrix}$ is a positive definite matrix.
\begin{flalign*}
&\textmd{Hence, } (\boldsymbol{\tilde{\xi_1}} - \boldsymbol{\xi_1^0})(\sqrt{n}\mathbf{D})^{-1} =  \left[-\frac{1}{\sqrt{n}} \mathbf{I_1}'(\boldsymbol{\xi_1^0})\mathbf{D}\right] \left[\mathbf{D}\mathbf{I_1}''(\boldsymbol{\bar{\xi_1}})\mathbf{D}\right]^{-1} \xrightarrow{a.s.} 0.&
\end{flalign*}
\end{proof}

\justify
\emph{Proof of Theorem 4:}
First let us prove the consistency of the estimates of the non-linear parameters of the first component of the multiple component model, that is, $\alpha_1^0$ and $\beta_1^0$. For notational simplicity we assume $p = 2$. \\ \\
Thus, $y(t) = A_1^0 \cos(\alpha_1^0 t + \beta_1^0 t^2) + B_1^0 \sin(\alpha_1^0 t + \beta_1^0 t^2) +  A_2^0 \cos(\alpha_2^0 t + \beta_2^0 t^2) + B_2^0 \sin(\alpha_2^0 t + \beta_2^0 t^2) + X(t).$\\ \\
Consider:  $\frac{1}{n}\bigg(I_1(\alpha,\beta) - I_1(\alpha_1^0,\beta_1^0)\bigg)$
\begin{flalign*}
&  = \frac{1}{n^2}\Bigg[\Bigg|\sum_{t=1}^{n}y(t)e^{-i(\alpha t + \beta t^2)}\Bigg| - \Bigg|\sum_{t=1}^{n}y(t)e^{-i(\alpha_1^0 t + \beta_1^0 t^2)}\Bigg|\Bigg] & \\
& = \frac{1}{n^2}\Bigg[\Bigg\{\sum_{t=1}^{n}y(t)\cos(\alpha t + \beta t^2)\Bigg\}^2 + \Bigg\{\sum_{t=1}^{n}y(t)\sin(\alpha t + \beta t^2)\Bigg\}^2 - \Bigg\{\sum_{t=1}^{n}y(t)\cos(\alpha_1^0 t + \beta_1^0 t^2)\Bigg\}^2 & \\
& \quad - \Bigg\{\sum_{t=1}^{n}y(t)\sin(\alpha_1^0 t + \beta_1^0 t^2)\Bigg\}^2\Bigg]. &
\end{flalign*}
The set  $S_c = \{ (\alpha , \beta): |\alpha - \alpha_1^0| > c \ \textmd{or} \  |\beta - \beta_1^0|  > c \} $ can be split into two parts and written as $S_c^1 \cup S_c^2$, where \\ \\
$S_c^1 = \{(\alpha,\beta): |\alpha - \alpha_1^0| > c  \textmd{ or }  |\beta - \beta_1^0| > c, \ (\alpha,\beta) = (\alpha_2^0, \beta_2^0)\}$, and \\ \\
$S_c^2 = \{(\alpha,\beta): |\alpha - \alpha_1^0| > c \textmd{ or } |\beta - \beta_1^0| > c, \ (\alpha,\beta)  \neq (\alpha_2^0, \beta_2^0)\}$. \\
\begin{flalign*}
&\limsup_{n \rightarrow \infty} \sup_{S_c^1}\frac{1}{n}\bigg(I_1(\alpha,\beta) - I_1(\alpha_1^0,\beta_1^0)\bigg)&\\
& = \limsup_{n \rightarrow \infty}\sup_{S_c^1}\frac{1}{n^2} \Bigg[\Bigg\{\sum_{t=1}^{n}\bigg(\sum_{k=1}^2 A_k^0 \cos(\alpha_k^0 t + \beta_k^0 t^2) + B_k^0 \sin(\alpha_k^0 t + \beta_k^0 t^2) + X(t)\bigg)\cos(\alpha t + \beta t^2)\Bigg\}^2\Bigg]&\\
& + \limsup_{n \rightarrow \infty}\sup_{S_c^1} \frac{1}{n^2} \Bigg[\Bigg\{\sum_{t=1}^{n}\bigg(\sum_{k=1}^2 A_k^0 \cos(\alpha_k^0 t + \beta_k^0 t^2) + B_k^0 \sin(\alpha_k^0 t + \beta_k^0 t^2) + X(t)\bigg) \sin(\alpha t + \beta t^2)\Bigg\}^2\Bigg]&\\
& - \limsup_{n \rightarrow \infty}\sup_{S_c^1} \frac{1}{n^2} \Bigg[\Bigg\{\sum_{t=1}^{n}\bigg(\sum_{k=1}^2 A_k^0 \cos(\alpha_k^0 t + \beta_k^0 t^2) + B_k^0 \sin(\alpha_k^0 t + \beta_k^0 t^2) + X(t)\bigg) \cos(\alpha_1^0 t + \beta_1^0 t^2)\Bigg\}^2\Bigg]&\\
& - \limsup_{n \rightarrow \infty}\sup_{S_c^1}\frac{1}{n^2} \Bigg[\Bigg\{\sum_{t=1}^{n}\bigg(\sum_{k=1}^2 A_k^0 \cos(\alpha_k^0 t + \beta_k^0 t^2) + B_k^0 \sin(\alpha_k^0 t + \beta_k^0 t^2) + X(t)\bigg) \sin(\alpha_1^0 t + \beta_1^0 t^2)\Bigg\}^2\Bigg]&\\
& = \frac{1}{4}({A_2^0}^2 + {B_2^0}^2 -{A_1^0}^2 - {B_1^0}^2) < 0 \ \ a.s.  \ (\textmd{Assumption} \ 3.)&.\\
& \textmd{Similarly, } \limsup_{n \rightarrow \infty}\sup_{S_c^2}\frac{1}{n}\bigg(I_1(\alpha,\beta) - I_1(\alpha_1^0,\beta_1^0)\bigg) =  \frac{1}{4}(0 + 0 - {A_1^0}^2 - {B_1^0}^2) < 0 \ \ a.s. &
\end{flalign*}
Therefore, $\tilde{\alpha_1}$ $\xrightarrow{a.s.}$ $\alpha_1^0$ and $\tilde{\beta_1}$ $\xrightarrow{a.s.}$ $\beta_1^0$ by Lemma~\ref{lemma:6}. Now we prove the consistency of linear parameter estimators $\tilde{A_1}$ and $\tilde{B_1}$. Observe that
\begin{equation*}
\tilde{A_1} = \frac{2}{n}\sum_{t=1}^{n} y(t)\cos(\tilde{\alpha_1} t + \tilde{\beta_1} t^2) = \frac{2}{n}\sum_{t=1}^{n} \bigg(\sum_{k=1}^{p}A_k^0\cos(\alpha_k^0 t + \beta_k^0 t^2) + B_k^0\sin(\alpha_k^0 t + \beta_k^0 t^2) + X(t)\bigg)\cos(\tilde{\alpha_1} t + \tilde{\beta_1} t^2).
\end{equation*}
\setlength{\belowdisplayskip}{0pt} \setlength{\belowdisplayshortskip}{0pt}
\setlength{\abovedisplayskip}{0pt} \setlength{\abovedisplayshortskip}{0pt}
\justify
We know that, $\frac{2}{n}\sum\limits_{t=1}^{n} X(t)\cos(\tilde{\alpha_1} t + \tilde{\beta_1} t^2) \rightarrow 0$. Now expanding $\cos(\tilde{\alpha_1} t + \tilde{\beta_1} t^2)$ by multivariate Taylor series around $(\alpha^0 , \beta^0)$ and using Lemmas~\ref{lemma:7} and~\ref{lemma:2}, we get:
\begin{comment}
\begin{flalign*}
\tilde{A_1} & \sim  \frac{2}{n}\sum_{t=1}^{n} \bigg(\sum_{k=1}^{p}A_k^0\cos(\alpha_k^0 t + \beta_k^0 t^2) + B_k^0\sin(\alpha_k^0 t + \beta_k^0 t^2)\bigg)\bigg[\cos(\alpha_1^0 t + \beta_1^0 t^2) + \begin{pmatrix} 
\frac{\partial \cos(\alpha_1^0 t + \beta_1^0 t^2)}{\partial \alpha} & \frac{\partial \cos(\alpha_1^0 t + \beta_1^0 t^2)}{\partial \beta}
\end{pmatrix} & \\
& \quad \begin{pmatrix} 
\tilde{\alpha_1} - \alpha_1^0 \\ 
\tilde{\beta_1} - \beta_1^0
\end{pmatrix}\bigg] &\\
& = \frac{2}{n}\sum_{t=1}^{n} \bigg(A_1^0\cos(\alpha_1^0 t + \beta_1^0 t^2) + B_1^0\sin(\alpha_1^0 t + \beta_1^0 t^2) + \sum_{k=2}^{p}A_k^0\cos(\alpha_k^0 t + \beta_k^0 t^2) + B_k^0\sin(\alpha_k^0 t + \beta_k^0 t^2)\bigg)&\\
& \quad \bigg[\cos(\alpha_1^0 t + \beta_1^0 t^2) - t (\tilde{\alpha_1} - \alpha_1^0)\sin(\alpha_1^0 t + \beta_1^0 t^2) - t ^2 (\tilde{\beta_1} - \beta_1^0)\sin(\alpha_1^0 t + \beta_1^0 t^2)\bigg] \rightarrow A_1^0 \ a.s.&
\end{flalign*}
 using Lemma~\ref{lemma:7} and the trigonometric identities in Lemma~\ref{lemma:2}. Similarly, $\tilde{B_1}$ $\rightarrow$ $B_1^0$ a.s.\\
 \end{comment}
$\tilde{A_1}$ $\rightarrow$ $A_1^0$ a.s. and $\tilde{B_1}$ $\rightarrow$ $B_1^0$ a.s.\\ \qed \\ 
\justify
\emph{Proof of Theorem 5:}
From Theorem~\ref{theorem:4} and Lemmas~\ref{lemma:6} and~\ref{lemma:7}, we have the following: 
\begin{flalign*}
&\tilde{A_1} = A_1^0 + o(1), & \\
&\tilde{B_1} = B_1^0 + o(1), & \\
&\tilde{\alpha_1} = \alpha_1^0 + o(n^{-1}), & \\
&\tilde{\beta_1} = \beta_1^0 + o(n^{-2}). &
\end{flalign*}
\begin{flalign}\label{eq:16}
&\textmd{Thus,} \ \tilde{A_1} \cos(\tilde{\alpha_1} t + \tilde{\beta_1} t^2) + \tilde{B_1} \sin(\tilde{\alpha_1} t + \tilde{\beta_1} t^2) = A_1^0 \cos(\alpha_1^0 t + \beta_1^0 t^2) + B_1^0  \sin(\alpha_1^0 t + \beta_1^0 t^2) + o(1).&
\end{flalign}
Now let us consider the difference $\frac{1}{n}\bigg(I_2(\alpha, \beta) - I_2(\alpha_2^0,  \beta_2^0)\bigg)$ \\
\begin{equation}\begin{split}\label{eq:17}
& = \frac{1}{n^2}\Bigg[\Bigg\{\sum_{t=1}^{n}y^1(t)\cos(\alpha t + \beta t^2)\Bigg\}^2 + \Bigg\{\sum_{t=1}^{n}y^1(t)\sin(\alpha t + \beta t^2)\Bigg\}^2 - \Bigg\{\sum_{t=1}^{n}y^1(t)\cos(\alpha_2^0 t + \beta_2^0 t^2)\Bigg\}^2\\
& \quad - \Bigg\{\sum_{t=1}^{n}y^1(t)\sin(\alpha_2^0 t + \beta_2^0 t^2)\Bigg\}^2\Bigg].
\end{split} \end{equation}
Here, $y^1(t) = y(t) - \tilde{A_1} \cos(\tilde{\alpha_1} t + \tilde{\beta_1} t^2) + \tilde{B_1} \sin(\tilde{\alpha_1} t + \tilde{\beta_1} t^2)$, that is the new data obtained by removing the effect of the first component from the observed data $y(t)$. Using~(\ref{eq:16}), we have $$y^1(t) = o(1) + \sum_{k=2}^p (A_k^0 \cos(\alpha_k^0 t + \beta_k^0 t^2) + B_k^0 \sin(\alpha_k^0 t + \beta_k^0 t^2)) + X(t).$$
Substituting this in~(\ref{eq:17}), we have:
\begin{flalign*}
&\frac{1}{n}\bigg(I_2(\alpha,\beta) - I_2(\alpha_2^0,\beta_2^0)\bigg)& \\
& = \frac{1}{n^2} \Bigg[\Bigg\{\sum_{t=1}^{n}\bigg( o(1) + \sum_{k=2}^p A_k^0 \cos(\alpha_k^0 t + \beta_k^0 t^2) + B_k^0 \sin(\alpha_k^0 t + \beta_k^0 t^2) + X(t)\bigg)\cos(\alpha t + \beta t^2)\Bigg\}^2\Bigg]&\\
& + \frac{1}{n^2} \Bigg[\Bigg\{\sum_{t=1}^{n}\bigg( o(1) + \sum_{k=2}^p A_k^0 \cos(\alpha_k^0 t + \beta_k^0 t^2) + B_k^0 \sin(\alpha_k^0 t + \beta_k^0 t^2) + X(t)\bigg)\sin(\alpha t + \beta t^2)\Bigg\}^2\Bigg]&\\
&- \frac{1}{n^2} \Bigg[\Bigg\{\sum_{t=1}^{n}\bigg( o(1) + \sum_{k=2}^p A_k^0 \cos(\alpha_k^0 t + \beta_k^0 t^2) + B_k^0 \sin(\alpha_k^0 t + \beta_k^0 t^2) + X(t)\bigg)\cos(\alpha_2^0 t + \beta_2^0 t^2)\Bigg\}^2\Bigg]&\\
&- \frac{1}{n^2} \Bigg[\Bigg\{\sum_{t=1}^{n}\bigg( o(1) + \sum_{k=2}^p A_k^0 \cos(\alpha_k^0 t + \beta_k^0 t^2) + B_k^0 \sin(\alpha_k^0 t + \beta_k^0 t^2) + X(t)\bigg)\sin(\alpha_2^0 t + \beta_2^0 t^2)\Bigg\}^2\Bigg].&\\
\end{flalign*}
\justify
The set  $S_c = \{ (\alpha , \beta): |\alpha - \alpha_2^0| > c \ or \  |\beta - \beta_2^0|  > c \} $ can be split into $p$ sets and written as $S_c^1 \cup S_c^2 \cup \cdots \cup S_c^p$ where \\ \\
$S_c^1 = \{(\alpha,\beta): |\alpha - \alpha_2^0| > c  \ or\  |\beta - \beta_2^0| > c, \ (\alpha,\beta) = (\alpha_1^0, \beta_1^0)\}$, \\ \\
$S_c^2 = \{(\alpha,\beta): |\alpha - \alpha_2^0| > c  \ or\  |\beta - \beta_2^0| > c, \ (\alpha,\beta) = (\alpha_3^0, \beta_3^0)\}$, \\ \\
$\vdots$\\ \\
$S_c^{p-1} = \{(\alpha,\beta): |\alpha - \alpha_2^0| > c  \ or\  |\beta - \beta_2^0| > c, \ (\alpha,\beta) = (\alpha_p^0, \beta_p^0)\},$ \\ \\
$S_c^p= \{(\alpha,\beta): |\alpha - \alpha_2^0| > c  \ or\  |\beta - \beta_2^0| > c, \ (\alpha,\beta) \neq (\alpha_k^0, \beta_k^0),\textmd{ for  any } k = 1, \cdots, p\}.$
\begin{flalign*}
&\textmd{It can be easily seen that: } \limsup_{n \rightarrow \infty}\sup_{S_c^k}\frac{1}{n}\bigg(I_2(\alpha,\beta) - I_2(\alpha_2^0,\beta_2^0)\bigg) < 0 \ \  a.s.  \ \forall \ k = 1, \cdots, p.&
\end{flalign*}
Combining, we have $\limsup\limits_{n \rightarrow \infty}\sup\limits_{S_c}\frac{1}{n}\bigg(I_2(\alpha,\beta) - I_2(\alpha_2^0,\beta_2^0)\bigg)$ $<$ 0 $a.s.$
Therefore, $\tilde{\alpha_2}$ $\xrightarrow{a.s.}$ $\alpha_2^0$ and $\tilde{\beta_2}$ $\xrightarrow{a.s.}$ $\beta_2^0$ by Lemma~\ref{lemma:6}. Following the same argument as in Theorem~\ref{theorem:4}, we can prove the consistency of linear parameter estimators $\tilde{A_2}$ and $\tilde{B_2}$. \\ \qed

\justify
\emph{Proof of Theorem 6:} We know that the ALSEs of the linear parameters $A$ and $B$ are given by:
\begin{flalign*}
&\tilde{A} = \frac{2}{n} \sum_{t=1}^{n} y(t) \cos(\tilde{\alpha} t + \tilde{\beta} t^2), and &\\
&\tilde{B} = \frac{2}{n} \sum_{t=1}^{n} y(t) \sin(\tilde{\alpha} t + \tilde{\beta} t^2).&
\end{flalign*}
Let $k = p+1$, then 
$ \tilde{A}_{p+1} = \frac{2}{n} \sum\limits_{t=1}^{n} y^{p+1}(t) \cos(\tilde{\alpha}_{p+1} t + \tilde{\beta}_{p+1}t^2) $
where $\tilde{\alpha}_{p+1}$ and $\tilde{\beta}_{p+1}$ are obtained by maximising $I_{p+1}(\alpha,\beta)$ and $y^{p+1}(t) = y(t) - \bigg(\sum\limits_{k=1}^{p} \tilde{A}_k \cos(\tilde{\alpha}_{k} t + \tilde{\beta}_{k}t^2) + \tilde{B}_k \cos(\tilde{\alpha}_{k} t + \tilde{\beta}_{k}t^2)\bigg)$. \\
Using~(\ref{eq:16}), we have:
\begin{equation*}\begin{split}
\tilde{A}_{p+1} & = \frac{2}{n}\sum_{t=1}^{n}X(t) \cos(\tilde{\alpha}_{p+1} t + \tilde{\beta}_{p+1} t^2) + o(1),
\end{split}\end{equation*}
\begin{equation*}\begin{split}
\tilde{B}_{p+1} & = \frac{2}{n}\sum_{t=1}^{n}X(t) \sin(\tilde{\alpha}_{p+1} t + \tilde{\beta}_{p+1} t^2) + o(1).
\end{split}\end{equation*}
Using Lemma~\ref{lemma:1}, we have $\tilde{A}_{p+1} \xrightarrow{a.s.} 0 $ and $\tilde{B}_{p+1} \xrightarrow{a.s.} 0 $ \\ \qed 
\section*{Appendix D}\label{appendix:D}
\emph{Proof of Theorem 7:} First we prove that $(\boldsymbol{\tilde{\theta}_{1}} - \boldsymbol{\theta_1^0})\mathbf{D}^{-1}$ has the same asymptotic distribution as  $(\boldsymbol{\hat{\theta}_{1}} - \boldsymbol{\theta_1^0})\mathbf{D}^{-1}$. Consider $Q_1(\boldsymbol{\theta})$ as the error sum of squares, then 
\begin{flalign*}
\frac{1}{n}Q_1(\boldsymbol{\theta})& = \frac{1}{n}\sum_{t=1}^{n} \bigg(y(t) - (A\cos(\alpha t + \beta t^2) + B \sin(\alpha t + \beta t^2))\bigg)^2 &\\
& = \frac{1}{n}\sum_{t=1}^{n}y(t)^2 - \frac{2}{n}\sum_{t=1}^{n} y(t)\{A\cos(\alpha t + \beta t^2) + B \sin(\alpha t + \beta t^2)\} + \frac{1}{2}\bigg(A^2 + B^2\bigg) + o(1) &\\
& = C - \frac {1}{n}J_1(\boldsymbol{\theta}) + o(1), \textmd{ where}&
\end{flalign*}
\begin{flalign*}
&C = \frac{1}{n} \sum_{t=1}^{n}y^2(t),\textmd{ and  } \frac {1}{n}J_1(\boldsymbol{\theta}) = \frac{2}{n}\sum_{t=1}^{n} y(t)\{A\cos(\alpha t + \beta t^2) + B \sin(\alpha t + \beta t^2)\}  - \frac{A^2 + B^2}{2}. &
\end{flalign*}
Working on the similar lines as for the one component model in  \nameref{appendix:B}, one can show that: \\
\begin{flalign*}
&\lim_{n \rightarrow \infty} \mathbf{Q_1}'(\boldsymbol{\theta_1^0}) \mathbf{D} =  \lim_{n \rightarrow \infty} -\mathbf{J_1}'(\boldsymbol{\theta_1^0}) \mathbf{D} \quad \quad \textmd{and}& 
&\lim_{n \rightarrow \infty} \mathbf{D}\mathbf{Q_1}''(\boldsymbol{\theta_1^0})\mathbf{D} = \boldsymbol{\Sigma_1}^{-1} = -\lim_{n \rightarrow \infty} \mathbf{D}\mathbf{J_1}''(\boldsymbol{\theta_1^0})\mathbf{D}.&
\end{flalign*}
\justify
Since at $(\tilde{A} , \tilde{B} , \alpha , \beta)$,  $J_1(\tilde{A} , \tilde{B} , \alpha , \beta)$ =  $I_1(\alpha , \beta)$,
the estimator of $\boldsymbol{\theta_1^0}$ which maximises $J_1(\boldsymbol{\theta})$ is equivalent to $\boldsymbol{\tilde{\theta_1}}$, the ALSE of $\boldsymbol{\theta_1^0}$ which implies $\mathbf{J_1}'(\boldsymbol{\tilde{\theta_1}})$ = 0. Now expanding $\mathbf{J_1}'(\boldsymbol{\tilde{\theta_1}})$ around $\mathbf{J_1}'(\boldsymbol{\theta_1^0})$ by Taylor series, we have:
\begin{flalign*}
&(\boldsymbol{\tilde{\theta_1}} - \boldsymbol{\theta_1^0}) = -\mathbf{J_1}'(\boldsymbol{\theta_1^0})[\mathbf{J_1}''(\boldsymbol{\bar{\theta_1}})]^{-1}.&\\
\Rightarrow &(\boldsymbol{\tilde{\theta_1}} - \boldsymbol{\theta_1^0})\mathbf{D}^{-1} = -[\mathbf{J_1}'(\boldsymbol{\theta_1^0})\mathbf{D}][\mathbf{D}\mathbf{J_1}''(\boldsymbol{\bar{\theta_1}})\mathbf{D}]^{-1}.&
\end{flalign*}
Also, $\lim\limits_{n \rightarrow \infty}[\mathbf{D}\mathbf{J_1}''(\boldsymbol{\bar{\theta_1}})\mathbf{D}] = \lim\limits_{n \rightarrow \infty}[\mathbf{D}\mathbf{J_1}''(\boldsymbol{\theta_1^0})\mathbf{D}].$
\begin{flalign*}
&\textmd{Therefore }\lim_{n \rightarrow \infty}(\boldsymbol{\tilde{\theta_1}} - \boldsymbol{\theta_1^0})\mathbf{D}^{-1}  =  \lim_{n \rightarrow \infty}(\boldsymbol{\hat{\theta_1}} - \boldsymbol{\theta_1^0})\mathbf{D}^{-1}.&
\end{flalign*}
It follows that they have the same asymptotic distribution.\\
Next we prove that $(\boldsymbol{\tilde{\theta_2}} - \boldsymbol{\theta_2^0})\mathbf{D}^{-1}$ has the same asymptotic distribution as that of $(\boldsymbol{\hat{\theta_2}} - \boldsymbol{\theta_2^0})\mathbf{D}^{-1}$.
The estimates of the parameters associated with the second component are obtained by minimising the following:\\
$$\frac{1}{n}Q_2(\boldsymbol{\theta})  = \frac{1}{n}\sum_{t=1}^{n} \bigg(y^1(t) - (A\cos(\alpha t + \beta t^2) + B \sin(\alpha t + \beta t^2))\bigg)^2. $$
\begin{flalign*}
&\textmd{Here }  y^1(t) = y(t) - \tilde{A_1}\cos(\tilde{\alpha_1}t + \tilde{\beta_1} t^2) - \tilde{B_1}\sin(\tilde{\alpha_1}t + \tilde{\beta_1} t^2).&\\
&\textmd{Also, } \frac{1}{n}Q_2(\boldsymbol{\theta}) = C_1 - \frac {1}{n}J_2(\boldsymbol{\theta}) + o(1),&\\
&\textmd{where, }
C_1 = \frac{1}{n} \sum_{t=1}^{n}y^1(t), \ and \ \ 
\frac {1}{n}J_2(\boldsymbol{\theta_2}) = \frac{2}{n}\sum_{t=1}^{n} y^1(t)\{A\cos(\alpha t + \beta t^2) + B \sin(\alpha t + \beta t^2)\}  - \frac{A^2 + B^2}{2}.&
\end{flalign*}
Proceeding in the similar way as for the first component, we compute $\frac{1}{n}\textbf{J}'_2(\boldsymbol{\theta})$ and $\frac{1}{n}\textbf{Q}'_2(\boldsymbol{\theta})$ at $\theta$ = $\boldsymbol{\theta_2^0}$ and we get:
$$\lim_{n \rightarrow \infty} \mathbf{Q_2}'(\boldsymbol{\theta_2^0}) \mathbf{D} =  \lim_{n \rightarrow \infty} -\mathbf{J_2}'(\boldsymbol{\theta_2^0}) \mathbf{D}.$$
Again at $(\tilde{A} , \tilde{B} , \alpha , \beta)$, $J_2(\tilde{A} , \tilde{B} , \alpha , \beta) = I_2(\alpha , \beta)$. Hence the estimator of $\boldsymbol{\theta_2^0}$ which maximises $J_2(\boldsymbol{\theta})$ is equivalent to $\boldsymbol{\tilde{\theta_2}}$, the ALSE of $\boldsymbol{\theta_2^0}$. \\
Thus, $\mathbf{J_2}'(\boldsymbol{\tilde{\theta_2}})$ = 0 and on expanding $\mathbf{J_2}'(\boldsymbol{\tilde{\theta_2}})$ around $\mathbf{J_2}'(\boldsymbol{\theta_2^0})$ by Taylor series expansion, we have:
\begin{flalign*}
&(\boldsymbol{\tilde{\theta_2}} - \boldsymbol{\theta_2^0}) = -\mathbf{J_2}'(\boldsymbol{\theta_2^0})[\mathbf{J_2}''(\boldsymbol{\bar{\theta_2}})]^{-1}.&\\
\Rightarrow &(\boldsymbol{\tilde{\theta_2}} - \boldsymbol{\theta_2^0})\mathbf{D}^{-1} = -[\mathbf{J_2}'(\boldsymbol{\theta_2^0})\mathbf{D}][\mathbf{D}\mathbf{J_2}''(\boldsymbol{\bar{\theta_2}})\mathbf{D}]^{-1}. &
\end{flalign*}
It can be shown that, $\lim\limits_{n \rightarrow \infty}[\mathbf{D}\mathbf{J_2}''(\boldsymbol{\bar{\theta_2}})\mathbf{D}] = \lim\limits_{n \rightarrow \infty}[\mathbf{D}\mathbf{J_2}''(\boldsymbol{\theta_2^0})\mathbf{D}].$
One can compute $\lim\limits_{n \rightarrow \infty} \mathbf{D}\mathbf{J_2}''(\boldsymbol{\theta_2^0})\mathbf{D}$ and $\lim\limits_{n \rightarrow \infty}\mathbf{D}\mathbf{Q_2}''(\boldsymbol{\theta_2^0})\mathbf{D}$, and see that $\lim\limits_{n \rightarrow \infty} \mathbf{D}\mathbf{J_2}''(\boldsymbol{\theta_2^0})\mathbf{D} = -\boldsymbol{\Sigma_2}^{-1} = -\lim\limits_{n \rightarrow \infty} \mathbf{D}\mathbf{Q_2}''(\boldsymbol{\theta_2^0})\mathbf{D}$ as defined in~(\ref{eq:12}). 
\justify
Thus $ (\boldsymbol{\tilde{\theta_2}} - \boldsymbol{\theta_2^0})\mathbf{D}^{-1}$ and $(\boldsymbol{\hat{\theta_2}} - \boldsymbol{\theta_2^0})\mathbf{D}^{-1}$ have the same asymptotic distribution. This result can be extended for all $k = 3, \cdots, p$ and the proofs follow exactly in the same manner. \\ \qed \\
\end{appendices}
\bibliographystyle{plainnat}

\begin{thebibliography}{30}

\bibitem[1986]{1986}  
Abatzoglou, T. J., 1986 \href{http://ieeexplore.ieee.org/stamp/stamp.jsp?arnumber=4104290}{"Fast maximnurm likelihood joint estimation of frequency and frequency rate."} IEEE Transactions on Aerospace and Electronic Systems, 6, pp. 708-715.

\bibitem[1960]{1960} 
Bello, P., 1960 \href{http://ieeexplore.ieee.org/stamp/stamp.jsp?arnumber=1057562} {"Joint estimation of delay, Doppler, and Doppler rate."} IRE Transactions on Information Theory, 6(3), pp. 330-341.

\bibitem[1999]{1999}
Olivier, B., Ghogho, M. and Swami, A., 1999 \href{http://ieeexplore.ieee.org/stamp/stamp.jsp?arnumber=806067} {"Parameter estimation for random amplitude chirp signals."} IEEE Transactions on Signal Processing, 47(12), pp. 3208-3219.

\bibitem[1990]{1990_1}
Djuric, P. M., and Kay, S. M., 1990 \href{http://ieeexplore.ieee.org/stamp/stamp.jsp?arnumber=61538} {"Parameter estimation of chirp signals."} IEEE Transactions on Acoustics, Speech, and Signal Processing, 38(12), pp. 2118-2126.

%\bibitem[2001]{2001}
%Flandrin, P. \href{http://perso.ens-lyon.fr/patrick.flandrin/SPIE01_PF.pdf} {"Time-frequency and chirps."} Proc. SPIE. Vol. 4391. 2001.

\bibitem[1997]{1997}
Ikram, M. Z., Abed-Meraim, K. and Hua, Y., 1997 \href{http://ac.els-cdn.com/S1051200497902864/1-s2.0-S1051200497902864-main.pdf?_tid=4430bc46-6ec2-11e7-ad19-00000aacb35f&acdnat=1500716810_80c53c229d30f17ef6e4d4812811e175}{ "Fast quadratic phase transform for estimating the parameters of multicomponent chirp signals."} Digital Signal Processing, 7(2), pp. 127-135.

\bibitem[1961]{1961} 
Kelly, E. J., 1961 \href{http://ieeexplore.ieee.org/stamp/stamp.jsp?arnumber=5008321}{"The radar measurement of range, velocity and acceleration."} IRE Transactions on Military Electronics, 1051(2), pp. 51-57.

%\bibitem[1987]{1987}
%Kumaresan, R. and Verma, S. "On estimating the parameters of chirp signals using rank reduction techniques." Proc. 21st Asilomar Conf. Signals, Syst., Comput. 1987.

\bibitem[2008]{2008}
Kundu, D. and Nandi, S., 2008 \href{http://home.iitk.ac.in/~kundu/paper117.pdf} {"Parameter estimation of chirp signals in presence of stationary noise."} Statistica Sinica, pp. 187-201.

\bibitem[2012]{2012}
Kundu, D. and Nandi, S., 2012 \href{http://www.springer.com/in/book/9788132206279} {Statistical Signal Processing: Frequency Estimation.} Springer Science \& Business Media.

\bibitem[2013]{2013}
Lahiri, A., Kundu, D. and Mitra, A., 2013 \href{https://link.springer.com/article/10.1007/s13571-012-0048-x}{ "Efficient algorithm for estimating the parameters of two dimensional chirp signal."},  Sankhya B, 75(1), pp. 65-89.

\bibitem[2014]{2014}
Lahiri, A., Kundu, D. and Mitra, A., 2014 \href{http://home.iitk.ac.in/~kundu/chirp-one-LAD-rev-2.pdf} {"On least absolute deviation estimators for one-dimensional chirp model."} Statistics, 48(2), pp. 405-420.

\bibitem[2015]{2015}
Lahiri, A., Kundu, D. and Mitra, A., 2015 \href{http://www.sciencedirect.com/science/article/pii/S0047259X15000329} {"Estimating the parameters of multiple chirp signals."} Journal of Multivariate Analysis, 139, pp. 189-206.

%\bibitem[1992]{1992}
%Liang, R. M. and Arun, K. S. \href{http://ieeexplore.ieee.org/stamp/stamp.jsp?arnumber=226517} {"Parameter estimation for superimposed chirp signals."} Acoustics, Speech, and Signal Processing, 1992. ICASSP-92., 1992 IEEE International Conference on. Vol. 5. IEEE, 1992.

%\bibitem[2000]{2000}
%Lin, C. C. and Djuric, P. M. \href{http://ieeexplore.ieee.org/stamp/stamp.jsp?arnumber=861938}{ "Estimation of chirp signals by MCMC."} Acoustics, Speech, and Signal Processing, 2000. ICASSP'00. Proceedings. 2000 IEEE International Conference on. Vol. 1. IEEE, 2000.

\bibitem[2016]{2016}
Mazumder, S., 2017 \href{http://www.tandfonline.com/doi/pdf/10.1080/03610918.2015.1053921?needAccess=true} {"Single-step and multiple-step forecasting in one-dimensional single chirp signal using MCMC-based Bayesian analysis."} Communications in Statistics-Simulation and Computation, 46(4), pp. 2529-2547.

\bibitem[2004]{2004}
Nandi, S. and Kundu, D., 2004 \href{http://home.iitk.ac.in/~kundu/paper91.pdf} {"Asymptotic properties of the least squares estimators of the parameters of the chirp signals."} Annals of the Institute of Statistical Mathematics, 56(3), pp. 529-544.

\bibitem[1991]{1991}
Peleg, S. and Porat, B., 1991 \href{http://ieeexplore.ieee.org/stamp/stamp.jsp?arnumber=85033}{ "Linear FM signal parameter estimation from discrete-time observations."} IEEE Transactions on Aerospace and Electronic Systems, 27(4), pp. 607-616.

\bibitem[1988]{1988}
Rice, J. A. and Rosenblatt, M., 1988 \href{http://www.jstor.org/stable/pdf/2336597.pdf?refreqid=excelsior\%3Aa43446ea2e53233ab811fcca4e024ea4} { "On frequency estimation."} Biometrika, 75(3), pp. 477-484.

\bibitem[1961]{1961_2}
Richards, F. SG., 1962 \href{https://www.researchgate.net/publication/268494756_A_method_of_Maximum_Likelihood_estimation}{"A method of maximum-likelihood estimation."} Journal of the Royal Statistical Society. Series B (Methodological), pp. 469-475.

\bibitem[2002]{2002}
Saha, S. and Kay, S. M., 2002 \href{http://ieeexplore.ieee.org/stamp/stamp.jsp?arnumber=978378} {"Maximum likelihood parameter estimation of superimposed chirps using Monte Carlo importance sampling."} IEEE Transactions on Signal Processing, 50(2), pp. 224-230.

\bibitem[1995]{1995}
Shamsunder, S., Giannakis, G. B. and Friedlander, B., 1995 \href{http://ieeexplore.ieee.org/stamp/stamp.jsp?arnumber=348131} {"Estimating random amplitude polynomial phase signals: A cyclostationary approach."} IEEE Transactions on Signal Processing, 43(2), pp. 492-505.

\bibitem[1954]{1954}
Vinogradov, I. M., 1954 \href{https://www.researchgate.net/publication/247931639_The_method_of_trigonometrical_sums_in_the_theory_of_numbers} {"The method of Trigonometrical Sums in the Theory of Numbers, translated from the Russian, revised and annotated by KF Roth and A."} Davenport, Interscience, London, pp. 8-10.

\bibitem[1971]{1971}
Walker, A. M., 1971 \href{https://www.jstor.org/stable/pdf/1426034.pdf?refreqid=excelsior\%3A6433f0247072ffc32eb7c37646cbc995} {"On the estimation of a harmonic component in a time series with stationary independent residuals."} Biometrika, 58(1), pp. 21-36.

\bibitem[2006]{2006}
Wang, P. and Yang, J., 2006 \href{https://www.researchgate.net/profile/Pu_Wang14/publication/220137449_Multicomponent_chirp_signals_analysis_using_product_cubic_phase_function/links/54aeb31b0cf29661a3d3a752.pdf} {"Multicomponent chirp signals analysis using product cubic phase function."} Digital Signal Processing, 16(6), pp. 654-669.

\bibitem[1952]{1952}
Whittle, P., 1952 "The simultaneous estimation of a time series harmonic components and covariance structure." Trabajos Estadistica, 3(1), pp. 43-57.
%\bibitem[1990]{1990_2}
%Montgomery, Hugh L. \href{https://www.researchgate.net/publication/243712263_Ten_Lectures_on_the_Interface_Between_Analytic_Number_Theory_and_Harmonic_Analysis?el=1_x_8&enrichId=rgreq-4c4e2131537374a646a10aa41f5eecd3-XXX&enrichSource=Y292ZXJQYWdlOzI3Mjg5MTYzOTtBUzoyMDU4MjQxOTIyNDU3NjBAMTQyNjA4MzcxOTkzOA==} {Ten lectures on the interface between analytic number theory and harmonic analysis.} No. 84. American Mathematical Soc., 1994.



















\end{thebibliography}

\end{document}